\documentclass[11pt]{article}

\usepackage[latin1]{inputenc}
\usepackage[T1]{fontenc}
\usepackage[english]{babel}
\usepackage{amsmath,amssymb,amsthm}
\usepackage{mathrsfs} 
\usepackage{graphicx}
\usepackage{bbm}
\usepackage{color}
\usepackage{comment}
\usepackage[textsize=small]{todonotes}       
\usepackage{color}
\usepackage{enumerate}

\newtheorem{theorem}{Theorem}[section]

\newtheorem{definition}[theorem]{Definition}
\newtheorem{lemma}[theorem]{Lemma}

\newtheorem{proposition}[theorem]{Proposition}

\newtheorem{remark}[theorem]{Remark}
\newtheorem{cond}[theorem]{Condition}

\newcommand{\be}{\begin{equation}}
\newcommand{\ee}{\end{equation}}
\newcommand{\bea}{\begin{align}}
\newcommand{\ea}{\end{align}}
\newcommand{\nn}{\nonumber}

\newcommand{\prob}{\mathbb P}
\newcommand{\expec}{\mathbb E}

\newcommand{\atanh}{{\rm atanh}}

\newcommand{\asinh}{{\rm asinh}}

\newcommand{\dint}{{\rm d}}
\newcommand{\bbR}{\mathbb{R}}
\newcommand{\bbN}{\mathbb{N}}

\newcommand{\sss}{\scriptscriptstyle}

\numberwithin{equation}{section}

\newcommand{\e}{{\mathrm e}}
\newcommand{\sN}{n}

\newcommand{\rem}[1]{}

\newcommand{\tb}{\tilde{\beta}}
\newcommand{\bft}{{\bf t}}
\newcommand{\bfs}{{\bf s}}

\newcommand{\cF}{{\cal F}}

\newcommand{\tS}{{\tt S}}
\newcommand{\tI}{{\tt I}}
\newcommand{\mwn}{m^{\sss(w)}_{\sN}}
\newcommand{\psiqe}{\psi^{\mathrm{qe}}}
\newcommand{\psian}{\psi^{\mathrm{an}}}
\newcommand{\betacqe}{\beta^{\mathrm{qe}}_c}
\newcommand{\betacan}{\beta^{\mathrm{an}}_c}
\newcommand{\muqe}{\mu^{\mathrm{qe}}}
\newcommand{\muan}{\mu^{\mathrm{an}}}
\newcommand{\Zqe}{Z^{\mathrm{qe}}}
\newcommand{\Zan}{Z^{\mathrm{an}}}
\newcommand{\ZICW}{Z^{\sss {\mathrm{ICW}}}}
\newcommand{\muICWt}{\mu^{\sss {\mathrm{ICW}}}_{n, {\mathbf{s}}}}

\def\1{{\mathchoice {1\mskip-4mu\mathrm l}      
{1\mskip-4mu\mathrm l}
{1\mskip-4.5mu\mathrm l} {1\mskip-5mu\mathrm l}}}
\newcommand{\indic}[1]{\1_{\{#1\}}}

\newcommand{\m}{\mu_{{\sN}}} 
 
\newcommand{\GRGnw}{\mathrm{GRG}_{\sN}(\boldsymbol{w})}

\newcommand{\col}[1]{\textcolor[rgb]{0,0,0}{#1}}

\newcommand{\eqn}[1]{\begin{equation}#1\end{equation}}
\newcommand{\eqan}[1]{\begin{align}#1\end{align}}

\newcommand{\convp}{\stackrel{\sss {\mathbb P}}{\longrightarrow}}
\newcommand{\convd}{\stackrel{\sss {\mathcal D}}{\longrightarrow}}

\newcommand{\E}{\mathbb E}
\newcommand{\R}{\mathbb R}

\newcommand{\bm}{{\bf m}}
\newcommand{\bp}{{\bf p}}
\newcommand{\bz}{{\bf z}}
\newcommand{\by}{{\bf y}}
\newcommand{\bJ}{{\bf J}}
\newcommand{\bt}{{\bf t}}
\newcommand{\bX}{{\bf X}}
\newcommand{\bx}{{\bf x}}
\newcommand{\bu}{{\bf u}}
\newcommand{\mnw}{{m^{\sss(w)}_{\sN}}}
\newcommand{\rnw}{{r^{\sss(w)}_{\sN}}}
\newcommand{\an}{\mathrm{an}}

\newcommand{\ch}[1]{\color{black}#1\color{black}}

\begin{document}
\title{Large deviations for the annealed Ising model\\on inhomogeneous random graphs: spins and degrees}
\author{
Sander Dommers
\footnote{Ruhr-University Bochum, Universit\"atsstra\ss e 150, 44780 Bochum, Germany. {\tt Sander.Dommers@ruhr-uni-bochum.de}}
\and
Cristian Giardin\`a
\footnote{University of Modena and Reggio Emilia, via G. Campi 213/b, 41125 Modena, Italy. {\tt cristian.giardina@unimore.it}}
\and
Claudio Giberti
\footnote{University of Modena and Reggio Emilia, via Amendola 2, 42122 Reggio Emilia, Italy. {\tt claudio.giberti@unimore.it}}
\and
Remco van der Hofstad
\footnote{Eindhoven University of Technology, P.O.\ Box 513, 5600 MB Eindhoven, The Netherlands. {\tt rhofstad@win.tue.nl}}
}

\maketitle

\pagenumbering{arabic}

\begin{abstract}
We prove a large deviations principle for the total spin and the number of edges under the annealed Ising measure on generalized random graphs. We also give detailed results on how the annealing over the Ising model changes the degrees of the vertices in the graph and show how it gives rise to interesting {\em correlated} random graphs.
\end{abstract}

\section{Introduction and main results}

Recently, there has been substantial work on Ising models on random graphs, as a paradigmatic model for dependent random variables on complex networks.
While much work exists on random graphs with {\em independent} randomness on the edges or vertices, such as percolation and first-passage percolation (see \cite{SaintFlour-vdH} for a substantial overview of results for these models on random graphs), the dependence of the random variables on the vertices raises many interesting new questions. We refer to 
\cite{Can17a, Can17b, DM, DomGiaGibHofPri16, DGH, DomGiaHof12, GGvdHP, GGvdHP2} for recent results on the Ising model on random graphs, as well as \cite[Chapter 5]{SaintFlour-vdH} and \cite{DM2} for overviews. The crux about the Ising model is that the variables that are assigned to the vertices of the random graph wish to be {\em aligned}, thus creating positive dependence. Since the Ising model lives on a random graph, we are dealing with non-trivial {\em double randomness}  of both the spin system as well as the random environment. While \cite{DM, DGH, DomGiaHof12, GGvdHP2} study the {\em quenched} setting, in which the random graph is either fixed (random-quenched) or the Boltzmann-Gibbs measure is averaged out with respect to the random medium (averaged-quenched), recently the {\em annealed} setting, in which both the partition function and the Boltzmann weight are averaged out separately has attracted substantial attention \cite{Can17a, Can17b, DomGiaGibHofPri16, GGvdHP}. The random graph models investigated are rank-1 inhomogeneous random graphs \cite{DomGiaGibHofPri16,GGvdHP}, as well as random regular graphs and configuration models  \cite{Can17a, Can17b, GGvdHP2}. Depending on the setting, the annealed setting may have a different critical temperature. However, as predicted by the non-rigorous physics work \cite{LeoVazVesZec02, DorGolMen08}, the annealed Ising model turns out to be in the same universality class as the quenched model for all settings investigated \cite{Can17b, DomGiaGibHofPri16, DomGiaHof12}.

In this paper, we extend the analysis of the annealed Ising model on inhomogeneous random graphs to their {\em large deviation properties}. We investigate both the large deviations of the total spin, which is a classical problem dating back at least to Ellis \cite{Ell06, Ellis}, but we also consider the large deviation properties under the annealed measure of purely graph quantities, such as the number of edges or the vertex degrees.  Such problems are in general difficult since the rate function is not convex at low temperatures ($\beta>\beta_c$), so the G\"artner-Ellis theorem cannot be used directly. 

Our main results provide a formula for the large deviation function of the total spin that holds true even when the hypothesis of the theorem are not satisfied, i.e., at {\em low} temperatures.
This formula is indeed valid for all values of the parameters determining the phase diagram. To overcome the lack of differentiability of the annealed pressure (which is a
necessary condition for the application of the G\"artner-Ellis theorem) at low temperatures, we shall use the key property that the annealed Ising model on the
generalized random graph can be mapped to an inhomogeneous mean-field (Curie-Weiss) model. As a consequence, the large deviation function of the total spin 
can be deduced from classical results for independent variables and application of the Varadhan's lemma.

The study of large deviations for the number of edges brings the fact to light that, if one focuses solely on graph observables
and properties, then annealing can be described in terms of a {\em modified law} for the graph.
Our results  show that in the annealed setting, the typical number of edges present is substantially larger than the typical value under the original law
of the graph, thus quantifying the effect that the annealing has on the structure of the random graph involved. 
As explained in more detail below, one could think of the annealed Ising model on a random graph as giving rise to a random graph with an 
interesting {\em correlation structure} between the edges. To gain more understanding on this correlation structure we also
investigate the degrees distribution under the annealed Ising measure. Again we find that the degree of a fixed vertex
(or the degree of a uniformly chosen vertex) under the modified graph law has a distribution with a larger mean.

\subsection{The annealed Ising model on generalized random graphs}
\label{sec-model}
We now introduce the model. We first define the specific random graph model, the so-called {\em generalized random graph}, and then define the (annealed) Ising model.

\subsubsection{Generalized random graph}
\label{sec-GRG}
To construct the {\em generalized random graph} \cite{BriDeiMar-Lof05}, let $I_{ij}$ denote the Bernoulli indicator that the edge between vertex $i$ and vertex $j$ is present and let  
$p_{ij} = \mathbb{P}\left(I_{ij} = 1\right)$ be the edge probability, where different edges are present independently. Further, consider a sequence of non-negative weights
$\boldsymbol{w} = (w_i)_{i \in [\sN]}$ whose label $i$ runs through the vertex set $[\sN]=\{1,\ldots,\sN\}$.
Then, the generalized random graph,  denoted by $\GRGnw$,
is defined by 
	\be
	\label{eq-prob-edge}
	p_{ij} \, = \, \frac{w_i w_j}{\ell_{\sN} + w_i w_j},
	\ee
where $\ell_{\sN}= \sum_{i\in[n]} w_i$ is the total weight of all vertices. 
Denote the law of $\GRGnw$ by $\prob$ and its expectation by $\expec$. 
There are many related random graph models (also called rank-1 inhomogeneous random graphs \cite{BolJanRio07}), such as the random graph with specified expected degrees or Chung-Lu model \cite{ChuLu02b, ChuLu06c} and the Poisson random graph or Norros-Reittu model \cite{NorRei06}. Janson \cite{Jans08a} shows that many of these models are asymptotically equivalent. Even though his results do not apply to the large deviation properties of these random graphs, all our results also apply to these other models.

We need to assume that the vertex weight sequences $\boldsymbol{w} = (w_i)_{i \in [\sN]}$ are sufficiently nicely behaved. 
Let $U_{\sN} \in [\sN]$ denote a uniformly chosen vertex in $\GRGnw$ and $W_{\sN} = w_{U_{\sN}}$ its weight. Then, the following condition defines the asymptotic weight $W$ and set the convergence properties of  $(W_{\sN})_{\sN\ge 1}$ to $W$:

\begin{cond}[Weight regularity]  There exists a random variable $W$  such that, as $\sN\rightarrow\infty$,
\label{cond-weightreg}
\begin{enumerate}[(a)]
\item $W_{\sN} \stackrel{\cal D}{\longrightarrow} W$, where $ \stackrel{\cal D}{\longrightarrow}$ denotes convergence in distribution;
\item $\mathbb{E}[W_{\sN}]=\frac{1}{\sN}\sum_{i\in [\sN]} w_{i} \rightarrow \mathbb{E}[W]< \infty$;
\item $\mathbb{E}[W_{\sN}^2]=\frac{1}{\sN}\sum_{i\in [\sN]} w_{i}^2 \rightarrow \mathbb{E}[W^2]< \infty$;
\end{enumerate}
Further, we assume that $\mathbb{E}[W]>0$.
\end{cond}

\medskip

As explained in more detail in \cite[Chapter 6]{vdH}, conditions (a)--(b)  imply that the empirical degree distribution of the random graph converges to a mixed Poisson distribution 
with mixing distribution $W$, i.e., the proportion of vertices with degree $k$ is close to the probability that a Poisson random variable with random parameter $W$ equals $k$.\\
We note also that, by uniform integrability,  Condition~\ref{cond-weightreg}(c)  implies (b).\\ \\
\noindent 
{\bf Notation.} Throughout this paper, given a probability measure $\mu$ we denote by $\expec_\mu$ the average w.r.t.\ $\mu$. 
\subsubsection{Annealed Ising model}
\label{sec-AIM}
Let $\sigma=(\sigma_i)_{i\in[\sN]} \in \{-1,+1\}^{\sN} =: \Omega_{\sN}$ be a spin configuration. Then, for a given graph $G_{\sN}=([\sN],E_{\sN})$, where $E_n  \subset [\sN] \times [\sN]$ denotes the edge set, the Ising model is defined by the following Boltzmann-Gibbs measure
	\be
	\label{eq-quenchedIsing}
	\muqe_{\sN}(\sigma) = \frac{1}{\Zqe_{\sN}(\beta,B)} \exp\left\{\beta \sum_{(i,j)\in E_{\sN}} \sigma_i \sigma_j + B\sum_{i\in[\sN]} \sigma_i\right\},
	\ee
where
	\be
	\Zqe_{\sN}(\beta,B) = \sum_{\sigma\in \Omega_{\sN}} \exp\Big\{\beta \sum_{(i,j)\in E_{\sN}} \sigma_i \sigma_j + B\sum_{i\in[\sN]} \sigma_i\Big\}\nn
	\ee
is the quenched partition function. Here $\beta\geq0$ is the inverse temperature and $B\in \mathbb{R}$ is the external field. When $G_{\sN}$ is a random graph, this is known as the {\em random quenched} Ising model~\cite{GGvdHP2}.

To obtain the {\em annealed model}, we take expectations with respect to the random graph measure in both the numerator and denominator of~\eqref{eq-quenchedIsing}, i.e., we define the {\em annealed} Ising measure by
	\be
	\label{eq-annealedIsing}
	\muan_{\sN}(\sigma) = \frac{\expec\biggl[\exp\biggl\{\beta \sum_{(i,j)\in E_{\sN}} \sigma_i \sigma_j + B\sum_{i\in[\sN]} \sigma_i\biggr\}\biggr]}{\Zan_{\sN}(\beta,B)} ,
	\ee
where the annealed partition function $\Zan_{\sN}(\beta,B)$ is equal to
	\be
	\Zan_{\sN}(\beta,B)= \expec[\Zqe_{\sN}(\beta,B)] =  \sum_{\sigma\in \Omega_{\sN}} \expec\biggl[\exp\biggl\{\beta \sum_{(i,j)\in E_{\sN}} \sigma_i \sigma_j \nn 
	+ B\sum_{i\in[\sN]} \sigma_i\biggr\}\biggr].
	\ee

\subsubsection{Previous results for the annealed Ising model on the generalized random graph}
\label{sec-previous-AIM}
In this section, we describe some important results about the annealed Ising model that have been derived previously.  An important quantity in the study of the annealed Ising model is the {\em annealed pressure} defined by
	\be
	\psian_{\sN}(\beta,B) = \frac{1}{\sN}\log  \Zan_{\sN}(\beta,B).\nn
	\ee
The thermodynamic limit of this quantity $\psian(\beta,B):=\lim_{\sN\to\infty} \psian_{\sN}(\beta,B)$ is determined in the following theorem:

\begin{theorem}[Annealed pressure \cite{GGvdHP}]\label{thm-annealedpressure}
Suppose that 
\col{Condition~\ref{cond-weightreg}
holds}. Then for all $0\le \beta < \infty$ and all $B \in \R$,
	\be
	\label{eq-annpressure}
	\psian(\beta,B) = \log 2 + \alpha(\beta) + \expec\biggl[\log\cosh\biggl(\sqrt{\frac{\sinh(\beta)}{\expec[W]}} W z^\star(\beta,B) + B\biggr)\biggr] - z^\star(\beta,B)^2/2,
	\ee
where $\alpha(\beta)=\lim_{\sN\to\infty} \alpha_{\sN}(\beta)$ with $\alpha_{\sN}(\beta)$  defined in~\eqref{eq-def-alpha} below is given by
	\be
	\alpha(\beta) = \tfrac{1}{2}(\cosh(\beta)-1)\mathbb{E}[W],
	\ee
and $z^\star(\beta,B)$ is, for $B\neq0$, given by the unique solution with the same sign as $B$ of the fixed-point equation
	\be
	\label{z-fixed-point}
	z=\expec\biggl[\tanh\biggl(\sqrt{\frac{\sinh(\beta)}{\expec[W]}} W z + B\biggr)\sqrt{\frac{\sinh(\beta)}{\expec[W]}} W\biggr],
	\ee
whereas for $B=0$, $z^\star(\beta,0) = \lim_{B\searrow 0}z^\star(\beta,B)$.
\end{theorem}
\medskip


\col{This theorem is proved in~\cite[Thm~1.1]{GGvdHP}. 
In Section~\ref{seq-LDP-weighted} we provide an alternative expression for the annealed pressure
that is instrumental for our large deviation analysis.}

In~\cite[Thm~1.1]{GGvdHP} it is also proved that the annealed Ising model on the generalized random graph has a second order 
phase transition at a critical inverse temperature $\betacan$ given by
\be
\betacan = \asinh\left(\frac{\mathbb{E}[W]}{\mathbb{E}[W^2]}\right)\;.
\ee
Denote by 
\be
S_{\sN} = \sum_{i\in[\sN]} \sigma_i,\nn
\ee
the {\em total spin}, and by
\be
M^{\mathrm{an}}_n(\beta,B)=\mathbb{E}_{\muan_n} \biggl(\frac{S_{\sN}}{{\sN}} \biggr),\nn
\ee
the {\em finite-volume annealed magnetization}.
 It is show in~\cite[Thms.~1.2,~1.3]{GGvdHP} that a strong law of large numbers (SLLN) and central limit theorem (CLT) holds for the total spin:
\begin{theorem}[SLLN and CLT \cite{GGvdHP}]
Suppose that Condition~\ref{cond-weightreg} (a)--(c) hold. Define the {\em uniqueness regime} of the parameters $(\beta,B)$ by
	\be
	\mathcal{U} = \left\{(\beta,B) \,:\, \beta\geq0, B\neq0 {\rm \ or\ } 0<\beta<\betacan, B=0 \right\},\nn
	\ee
and suppose that $(\beta,B)\in \mathcal{U} $. Then, for all $\varepsilon>0$ there exists a constant $L=L(\varepsilon)>0$ such that, for all $n$ sufficiently large,
	\be
	\mathbb{P}_{\muan_{\sN}} \biggl(\Bigl|\frac{1}{\sN} S_{\sN} - M^{\mathrm{an}} \Bigr|\biggr) \leq \e^{-\sN L},\nn
	\ee
where
	\be
	M^{\mathrm{an}}(\beta,B) =\expec\biggl[\tanh\biggl(\sqrt{\frac{\sinh(\beta)}{\expec[W]}} W z^\star(\beta,B) + B\biggr)\biggr],\nn
	\ee
being $z^\star(\beta,B)$  the solution of \eqref{z-fixed-point}, equals the annealed magnetization, that is $\lim_{n\to \infty} M^{\mathrm{an}}_n(\beta,B)$.\\
Furthermore,
	\be
	\frac{S_{\sN} - \mathbb{E}_{\muan_n} (S_{\sN})}{\sqrt{\sN}} \stackrel{\cal D}{\longrightarrow} \mathcal{N}(0,\chi^{\mathrm{an}}), \qquad {\rm w.r.t.\ }\quad \muan_{\sN} {\rm \ as\ } \sN\to\infty,\nn
	\ee
where $\chi^{\mathrm{an}}(\beta,B)=\frac{\partial}{\partial B} M^{\mathrm{an}}(\beta,B)$ is the annealed susceptibility and $\mathcal{N}(0,\sigma^2)$ denotes a centered normal random variable with variance $\sigma^2$.
\end{theorem}

Analogously one can define the random quenched pressure:
	\be
	\psiqe(\beta,B) = \lim_{\sN\to\infty} \psiqe_{\sN}(\beta,B) = \lim_{\sN\to\infty} \frac{1}{\sN} \log \Zqe_{\sN}(\beta,B).\nn
	\ee 
This has been determined for the GRG as well as other locally tree-like random graph models in~\cite{DM,DGH}, where it is also proven that $\psiqe(\beta,B)$ is a non-random quantity. An SLLN and CLT  for the total spin 
w.r.t.\ $\muqe_{\sN}$ have been obtained in~\cite{GGvdHP2}. In general, the quenched and annealed pressures are different, and also the critical temperatures of the models are different. The only exception that we are aware of is the random regular graph (see \cite{Can17a}). The critical temperature in the quenched setting will be denoted by 
$\betacqe$.

%

\subsection{Main results}
\label{sec-results}
In this paper, we study the spin sum in more detail (i.e. beyond the CLT scale) and prove a large deviation principle for $S_{\sN}$, \ch{as well as a weighted version that plays a crucial role in the annealed Ising model}. Let us start by recalling what a large deviation principle is. Given a sequence of random 
variables $(X_n)_{n\ge 1}$  taking values in the measurable space $({\cal X}, {\cal B})$, with ${\cal X}$ a topological space and ${\cal B}$ a $\sigma$-field of subsets of ${\cal X}$, then the large deviation principle is defined as follows:
\begin{definition}[Large deviation principle (LDP) \cite{DemZei09}]
We say that $(X_n)_{n\ge 1}$ satisfies an LDP with rate function $I(x)$ and speed $\sN$ w.r.t.\ a probability measure $\ch{(\prob_{\sN})_{n\geq 1}}$ if,
for all $F\in  {\cal B}$,
	\eqn{
	-\inf_{x\in F^o} I(x)\leq \liminf_{n\rightarrow \infty} \frac{1}{n} \log \prob_{\ch{\sN}}(X_n\in F)\leq \limsup_{n\rightarrow \infty} 
	\frac{1}{n} \log \prob_{\ch{\sN}}(X_n\in F)\leq -\inf_{x\in \bar{F}} I(x),\nn
	}
where $F^o$ denotes the interior of $F$ and $\Bar{F}$ its closure.
\end{definition}
\medskip

In this definition $I\colon {\cal X}\to [0, \infty]$ is a lower semicontinuous function. Our first main result is an LDP for the total spin in the high-temperature regime for both the random quenched and the annealed Ising model:
\begin{theorem}[\col{Total spin} LDPs in high-temperature regime]\label{thm-LDP-HT}
In the annealed Ising model, under 
\col{Condition~\ref{cond-weightreg}},
 the total spin $S_{\sN}$ satisfies an LDP w.r.t.\ $\mu^{\mathrm{an}}_{\sN}$ for $\beta\leq\betacan$ and $B\in \mathbb{R}$, with rate function
	\be\label{rate2}
	{{I}^{\mathrm{an}} }(x)=  \sup_{t}\left \{ x\, t - {\psian(\beta,B+t)}\right\}+ \psian(\beta,B).
	\ee
In the random quenched Ising model, 
\col{under Condition~\ref{cond-weightreg}},
the total spin $S_{\sN}$ also satisfies an LDP w.r.t.\ $\muqe_{\sN}$ for $\beta\leq\betacqe$  and $B\in \mathbb{R}$, with rate function
	\be
	{{I}^{\mathrm{qe}}}(x)=  \sup_{t}\left \{ x\, t - {\psiqe(\beta,B+t)}\right\}+ \psiqe(\beta,B).\nn
	\ee
\end{theorem}


\ch{The proof of Theorem \ref{thm-LDP-HT} is highly general, and applies to settings where the pressure is known to exist and to be differentiable. 
As such, the proof is basically identical for the annealed and quenched Ising models on $\GRGnw$.}

For the annealed Ising model we also prove an LDP for all positive temperatures. For this, we also introduce the total {\em weighted} spin 
	$$
	S^{\sss(w)}_{\sN}=\sum_{i\in [\sN]}  w_{i}\sigma_{i}.
	$$

\begin{theorem}[\col{Alternative \ch{form of the}{} pressure and} LDPs for the annealed Ising model]\label{thm-ldp-weighted}
For all $\beta\geq0$ and $B\in \mathbb{R}$, 
\col{under Condition~\ref{cond-weightreg}, the annealed pressure is given by
	\be
	\label{p-annealed-2dim}
         \psian(\beta,B) = -\inf_{\ch{(x_1,x_2)}} \left(I(x_1,x_2) - \frac{\sinh(\beta)}{2\E[W]}x_2^{2}-  B x_1 - \log 2 - \alpha(\beta)\right)
	\ee
	where
	\be
	I(x_1,x_2)= \sup_{(t_1,t_2)} \left ( t_1 x_1 + t_2 x_2 - \E [\log \cosh (t_1 + W t_2)]\right).\nn
	\ee
and the couple
}
$(S_{\sN}, S^{\sss(w)}_{\sN})$  satisfies an LDP w.r.t.\ $\muan_{\sN}$ with rate function
	\be
	\label{eq-jointLDPrate1}
	I^{\mathrm{an}}_{\beta,B}(x_1,x_2)=I(x_1,x_2) - \frac{\sinh(\beta)}{2\E[W]}x_2^{2}-  B x_1 - \log 2 - \alpha(\beta)+ \psian(\beta,B).
	\ee
\col{
Furthermore the annealed pressure has the alternative expression	
	\be
	\label{p-annealed-2dimB}
	\psian(\beta,B) = -\inf_{\ch{(x_1,x_2)}} \left(I^{\sss (B)}(x_1,x_2) -   \frac{\sinh(\beta)}{2\E[W]} x_2^{2} -\log\cosh B-\log 2-\alpha(\beta) \right),
	\ee
where
	\be
	I^{\sss (B)}(x_1,x_2) = \sup_{(t_1,t_2)} \left(t_1 x_1+t_2 x_2 - \E[ \log \cosh  (B + t_1 + W t_2)] \right) + \log \cosh B,\nn
	\ee
}		
and also with the alternative expression of the rate function given by
	\be\label{eq-jointLDPrate2}
	I^{{\mathrm{an}}\sss (B)}_{\beta,B}(x_1,x_2)= I^{\sss (B)}(x_1,x_2) -   \frac{\sinh(\beta)}{2\E[W]} x_2^{2} -\log\cosh B-\log 2-\alpha(\beta) + \psian(\beta,B).
	\ee
\end{theorem}
\medskip

Naturally, in the high-temperature setting, the large deviation rate functions in \eqref{rate2} and \eqref{eq-jointLDPrate1} (\col{or \eqref{eq-jointLDPrate2}}) coincide after the application
of a contraction principle. 
{\col{Combining Theorem \ref{thm-annealedpressure} and Theorem \ref{thm-ldp-weighted} we see that
the annealed pressure is either given by the optimization of a real function (as in \eqref{eq-annpressure}) or it can be expressed as the solution of a two-dimensional variational problem (as in \eqref{p-annealed-2dim} or \eqref{p-annealed-2dimB}). In Section \ref{seq-LDP-weighted} we shall prove Theorem \ref{thm-annealedpressure} starting from Theorem \ref{thm-ldp-weighted}, thus obtaining that the expressions
for the annealed pressure do coincide.}

We next discuss the LDP for 
the total number of edges in the annealed Ising model on $\GRGnw$:

\begin{theorem}[LDPs for the edges in the annealed Ising model]
\label{thm-LDP-edges}
\col{Suppose that \ch{Condition~\ref{cond-weightreg} holds.}} 
For all $\beta\geq0$ and $B\in \mathbb{R}$, the total number of edges $|E_{\sN}|$  satisfies an LDP w.r.t.\ $\muan_{\sN}$ with rate function that is the Legendre transform of 
the function which is explicitly computed in \eqref{ann_cgf_edges} below. 
Further, the number of edges under the annealed Ising model on $\GRGnw$ satisfies
	\be
	\label{number-edges-AIM}
	\frac{1}{\sN}|E_{\sN}|\convp \frac12 {z^\star(\beta,B)}^2  +  \frac 1 2 \cosh(\beta) \expec{[W]}.
	\ee
\end{theorem}
\medskip

We continue by investigating the limiting distribution of the degrees of vertices. Our main result is as follows:
\begin{theorem}[Degrees in the annealed Ising model]
\label{thm-degrees-AIM}
Suppose that \ch{Condition~\ref{cond-weightreg} holds}. For all $\beta\geq0$ and $B\in \mathbb{R}$, the moment generating function of the degree $D_j$ of vertex $j$ under $\muan_n$ satisfies
	\be
	\label{eq-MGF-degree}
 	\expec_{\muan_n}\big[\e^{tD_j}\big]=  (1+o(1)) \e^{\cosh(\beta) w_j (\e^t-1)}\frac{\cosh\Big(z^\star(\beta,B)\e^t w_j \sqrt{\frac{\sinh(\beta)}{\expec[W]}} +B \Big)} 
	{\cosh\Big(z^\star(\beta,B) w_j \sqrt{\frac{\sinh(\beta)}{\expec[W]}} +B \Big)}.
	\ee
Consequently, the degree $D_{U}$ of a uniformly chosen vertex satisfies
	\be
	\label{eq-MGF-degree-unif}
 	\lim_{\sN\rightarrow \infty} \expec_{\muan_n}\big[\e^{tD_{U}}\big]= \expec\bigg[\e^{\cosh(\beta) W (\e^t-1)}\frac{\cosh\Big(z^\star(\beta,B) \e^t W \sqrt{\frac{\sinh(\beta)}{\expec[W]}} +B \Big)} 
	{\cosh\Big(z^\star(\beta,B) W\sqrt{\frac{\sinh(\beta)}{\expec[W]}} +B \Big)}\bigg].
	\ee
\ch{In the above, $z^\star(\beta,B)$ is the solution to \eqref{z-fixed-point}.}
\end{theorem}
\medskip

We remark that in \eqref{eq-MGF-degree-unif} we both take the average w.r.t.\ the annealed measure $\muan_n$ {\em as well as} with the uniform vertex $U\in[n]$. 

\begin{remark}[Degree distribution annealed Ising model]
\label{rem-degree-ann-Ising}
We can restate \eqref{eq-MGF-degree-unif} as
	\eqn{
	\frac{1}{n}\sum_{v\in[n]} \expec_{\muan_n}\big[\e^{tD_{v}}\big]\rightarrow \expec\bigg[\e^{\cosh(\beta) W (\e^t-1)}\frac{\cosh\Big(z^\star(\beta,B) \e^t W \sqrt{\frac{\sinh(\beta)}{\expec[W]}} +B \Big)} 
	{\cosh\Big(z^\star(\beta,B) W\sqrt{\frac{\sinh(\beta)}{\expec[W]}} +B \Big)}\bigg].
	}
In \eqref{eq-MGF-degree}, we see that the moment generating function of a vertex having weight $w$ is close to 
	\[
	\e^{\cosh(\beta) w (\e^t-1)}\frac{\cosh\Big(z^\star(\beta,B)\e^t w \sqrt{\frac{\sinh(\beta)}{\expec[W]}} +B \Big)} 
	{\cosh\Big(z^\star(\beta,B) w\sqrt{\frac{\sinh(\beta)}{\expec[W]}} +B \Big)}.
	\]
We recognize $\e^{\cosh(\beta) w (\e^t-1)}$ as the moment generating function of a Poisson random variable with mean $\cosh(\beta) w$, which is multiplied by another function. However, this factor does not turn out to be a moment generating function.

By setting $a(\beta)=\sqrt{\frac{\sinh(\beta)}{\expec[W]}}$ for the sake of notation, we can rewrite the product of the second and third factors in the r.h.s.\ of \eqref{eq-MGF-degree} as
	$$
	\frac{ \e^{ (w_j+B) a(\beta) z^\star}   \e^{w_j (\cosh(\beta) + a(\beta)z^\star) (\e^t-1)} 
	+ \e^{  -(w_j+B) a(\beta) z^\star }   \e^{w_j (\cosh(\beta) - a(\beta)z^\star) (\e^t-1)}    }{2 \cosh\Big(\left (w_j+B \right ) a(\beta) z^\star  \Big)}.
	$$
This shows that the limiting moment generating function of $D_j$ is a mixed Poisson random variables with parameters  $w_j (\cosh(\beta) +Ya(\beta)z^\star)$,  where 
	\[
	\prob(Y=1)=1-\prob(Y=-1)=\frac{\e^{ (w_j+B) a(\beta) z^\star }}{2 \cosh\Big( \left (w_j+B \right ) a(\beta) z^\star \Big)},
	\]
provided $w_j (\cosh(\beta)\pm a(\beta)z^\star)$ are both positive. We lack a more detailed interpretation of the above two realizations.
\end{remark}
\medskip

Let us next relate Theorem \ref{thm-degrees-AIM} to Theorem \ref{thm-LDP-edges}.
We can use \eqref{eq-MGF-degree-unif} to show that, as in \eqref{number-edges-AIM}, 
	\eqn{
	\expec_{\muan_n} \big[\frac{1}{\sN}|E_{\sN}|\big]\rightarrow \frac 1 2 {z^\star(\beta,B)}^2  +  \frac 1 2 \cosh(\beta) \expec{[W]}.\nn
	}
Indeed, note that
	\eqn{
	\label{lim-number-edges}
	\expec_{\muan_n} \big[\frac{1}{\sN}|E_{\sN}|\big]=\tfrac{1}{2}\expec_{\muan_n}[D_{U}]=\frac12 \frac{d}{dt} \expec_{\muan_n}\big[\e^{tD_{U}}\big]\Big|_{t=0}.\nn
	}
Here, in the middle formula, we again take the average w.r.t.\ both $\muan_n$ as well as the uniform vertex $U\in[n]$.
Convergence of the moment-generating function implies convergence of all moments, so that
	\eqan{\label{lim-number-edges2}
	\lim_{n\rightarrow\infty} \expec_{\muan_n} \big[\frac{1}{\sN}|E_{\sN}|\big]&=\frac12\frac{d}{dt} \expec\bigg[\e^{\cosh(\beta) W (\e^t-1)}\frac{\cosh\Big(z^\star(\beta,B)\e^t W\sqrt{\frac{\sinh(\beta)}{\expec[W]}}+B\Big)} 
	{\cosh\Big(z^\star(\beta,B) W\sqrt{\frac{\sinh(\beta)}{\expec[W]}}+B\Big)}\bigg]\Big|_{t=0}\nn\\
	&=\frac 1 2 \cosh(\beta) \expec{[W]} +\frac12z^\star(\beta,B)\expec\bigg[W\sqrt{\frac{\sinh(\beta)}{\expec[W]}}\tanh\Big(z^\star(\beta,B) W\sqrt{\frac{\sinh(\beta)}{\expec[W]}}+B\Big)\bigg]\nn\\
	&=\frac 1 2 \cosh(\beta) \expec{[W]} + \frac12 z^\star(\beta,B)^2,
	}
as required, where we have made use of \eqref{z-fixed-point} in the last step. Thus, for \eqref{number-edges-AIM}, it suffices to prove that $\frac{1}{\sN}|E_{\sN}|$ is concentrated.
\qed 
\medskip

In the next theorem, we extend Theorem \ref{thm-degrees-AIM} to several vertices:
\begin{theorem}[Degrees of $m$ vertices in the annealed Ising model]
\label{thm-degrees-m-AIM}
Suppose that \ch{Condition~\ref{cond-weightreg} holds}. For all $\beta\geq0$ and $B\in \mathbb{R}$ and $m\in \mathbb{N}$, the moment generating function of the degrees $(D_1,D_2,\ldots, D_m)$ under $\muan_n$ satisfies
$$
\expec_{\muan_n}\Big[\e^{\sum_{i=1}^m t_i D_i} \Big] = \prod_{i=1}^m  \e^{ \cosh(\beta) w_i  (\e^{t_i}-1)  }\prod_{i=1}^m\frac{\cosh\Big(z^\star(\beta,B) \e^{t_i} w_i\sqrt{\frac{\sinh(\beta)}{\expec[W]}} +B \Big)} 
	{\cosh\Big(z^\star(\beta,B) w_i\sqrt{\frac{\sinh(\beta)}{\expec[W]}} +B \Big)} (1+o(1)).
$$
\end{theorem}
\medskip

Theorem \ref{thm-degrees-m-AIM} implies that the degrees of different vertices under the annealed measure are approximately {\em independent.}
	
\subsection{Discussion}
\label{sec-disc}
In this section, we discuss our results and state some further conjectures. 

\paragraph{Random-quenched LDP.}
For the random-quenched model we only obtain an LDP in the high-temperature regime. The difficulty in this analysis is that the rate function is non-convex at low temperature. This means that the usual technique relying on the G\"artner-Ellis theorem, by taking the Legendre transform of the cumulant generating function, does not work. The cumulant generating function can easily be expressed in terms of the difference of the pressure for different values of the external field $B$. However, this Legendre transform is the convex envelope of the cumulant  generating function. This raises the question how to do this for all inverse temperatures $\beta$.

\paragraph{Averaged-quenched LDP.}
The averaged quenched measure is defined as $\expec\big[\muqe_{\sN}(\sigma)\big]$ (recall \eqref{eq-quenchedIsing}). Here, even in the high-temperature regime, we are in trouble since the averaged quenched cumulant generating function is {\em not} a difference of pressures. Independently of the explicit computation, an interesting question is whether it is possible to relate the random-quenched and the averaged-quenched large deviation rate functions.


\paragraph{Large deviations of random graph quantities.}
As already mentioned in the introduction, if one is interested only in graph quantities, then the effect of the annealing
amounts to changing the graph law from $\mathbb{P}$ (the law of of $\GRGnw$) to a new law $\mathbb{P}_{\beta,B}$
depending on the two parameters $\beta$ and $B$. Evidently $\lim_{\beta\to 0, B\to 0} \mathbb{P}_{\beta,B} = \mathbb{P}$.
We know that under the law $\mathbb{P}$ a uniform degree has \ch{an asymptotic}{} mixed Poisson distribution with mixing distribution $W$.
From formula \eqref{eq-MGF-degree-unif} we see that in zero external field $B=0$, the moment generating function of a uniform 
degree changes in two ways:
firstly, in the high-temperature regime, the mixing distribution changes to  $W\cosh(\beta)$ (since $z^{\star}(\beta,0)=0$ there);
secondly, in the low-temperature region a new effect appears due to the non-zero value of $z^{\star}(\beta,0)$.
It would be of interest to invert the moment generating function \eqref{eq-MGF-degree-unif} and thus explicitly
characterize the distribution of a uniform degree at low temperatures. \ch{This can be done once we know that 
 $\cosh(\beta)-\sinh(\beta)z^\star(\beta,0)/\expec[W]$ is non-negative (see Remark \ref{rem-degree-ann-Ising}), but we do not know this to be true in general.
Also, as of yet, we have no interpretation for this novel mixed Poisson distribution for the degrees}.
It might also be interesting to investigate other properties of the random graph under the annealed Ising model.
\ch{An example would be the distribution of triangles, for which the positive dependence of edges enforced by the annealed Ising model might have a pronounced effect. 
A further interesting problem is to identify the large deviation rate function in a {\em joint} LDP for both the spin as  
well as the total number of edges.}

\paragraph{Organisation of this paper.} We start in Section~\ref{sec-overview} by describing an enlightening computation that is at the heart of our analysis. In Section~\ref{sec-LDP-spin}, we derive the LDP for the total spin \col{and the total weighted spin}. In Section~\ref{sec-LDP-edges}, we investigate the large deviation properties, as well as the weak convergence, of the number of edges in the annealed Ising model, thus quantifying the statement that under the annealed Ising model, there are more edges in the graph than for the typical graph. In Section~\ref{sec-degree-distr-AIM}, we investigate the degree distribution under the annealed Ising model. \col{Finally in the Appendix we re-derive the LDP for the total spin by combinatorial
arguments.}

\subsection{Preliminaries: an enlightening computation}
\label{sec-overview}

Our large deviations results are obtained from exact expressions for moment generating functions of spin or of edge variables under the annealed $\GRGnw$ measure.
Such exact expressions follow from the observation (already contained  in~\cite[Sec.~2.1]{GGvdHP}) that the annealed $\GRGnw$ measure can be
\ch{identified as an {\em inhomogeneous} Ising model on the complete graph, which is called the {\em rank-1 inhomogeneous Curie-Weiss model} in \cite{GGvdHP}.}{} In this paper, we will extend such computations significantly, for example by also including the edge statuses. We can write the numerator in the definition \eqref{eq-annealedIsing} of $\muan_{\sN}$ as
	\begin{align}
	\expec\biggl[\exp\biggl\{\beta \sum_{(i,j)\in E_{\sN}} \sigma_i \sigma_j + B\sum_{i\in[\sN]} \sigma_i\biggr\}\biggr]
	& =\expec\biggl[\exp\biggl\{\beta \sum_{1\leq i < j \leq {\sN}}I_{ij} \sigma_i \sigma_j + B\sum_{i\in[\sN]} \sigma_i\biggr\}\biggr]\nn\\
	&= \e^{B\sum_{i\in[\sN]} \sigma_i} \prod_{i<j} \expec\biggl[\e^{\beta I_{ij} \sigma_i \sigma_j }\biggr] \nn\\
	&= \e^{B\sum_{i\in[\sN]} \sigma_i} \prod_{i<j} \biggl[\e^{\beta \sigma_i \sigma_j }p_{ij} +1-p_{ij}\biggr],\nn
	\end{align}
where we have used the independence of the edges in the second equality. Define
	\be
	\label{def-betaij-Cij}
	\beta_{ij} = \frac12 \log \frac{1+p_{ij}(\e^\beta-1)}{1+p_{ij}(e^{-\beta}-1)}, \qquad {\rm and} \qquad C_{ij} =\frac{1+p_{ij}(\cosh(\beta)-1)}{\cosh(\beta_{ij})}.
	\ee
Then, we can write
	\be
	\e^{\beta \sigma_i \sigma_j }p_{ij} +1-p_{ij} = C_{ij}\e^{\beta_{ij}\sigma_i\sigma_j}.\nn
	\ee
Hence, also using the symmetry $\beta_{ij}=\beta_{ji}$,
	\begin{align}
	\expec\biggl[\exp\biggl\{\beta \sum_{(i,j)\in E_{\sN}} \sigma_i \sigma_j + B\sum_{i\in[\sN]} \sigma_i\biggr\}\biggr]
	&= G_n(\beta) \; \e^{\frac12 \sum_{i,j\in[\sN]} \beta_{ij} \sigma_i\sigma_j+B\sum_{i\in[\sN]} \sigma_i},\nn
	\end{align}
where 
$$
G_n(\beta) = \biggl(\prod_{1\le i<j \le n} C_{ij}\biggr) \biggl(\prod_{i\in[\sN]}\e^{-\beta_{ii}/2}\biggr)
$$
and
$\beta_{ii}$ is defined as in~\eqref{def-betaij-Cij} with $p_{ii} = w_i^2 / (\ell_{\sN}+w_i^2)$. Defining
	\be
	\label{eq-def-alpha}
	\alpha_{\sN}(\beta) = \frac{1}{\sN} \log  G_n(\beta)
	\ee
one has
	\be
	\expec\biggl[\exp\biggl\{\beta \sum_{(i,j)\in E_{\sN}} \sigma_i \sigma_j + B\sum_{i\in[\sN]} \sigma_i\biggr\}\biggr]
	= \e^{\sN \alpha_{\sN}(\beta)} \e^{\frac12 \sum_{i,j\in[\sN]} \beta_{ij} \sigma_i\sigma_j+B\sum_{i\in[\sN]} \sigma_i}.\nn
	\ee
We observe that  the quantity  $\e^{\frac12 \sum_{i,j\in[\sN]} \beta_{ij} \sigma_i\sigma_j+B\sum_{i\in[\sN]} \sigma_i}$ can be regarded as the Hamiltonian of an inhomogeneous Curie-Weiss model with couplings 
given by $(\beta_{ij})_{ij}$. Thus, the annealed Ising model on the  $\GRGnw$ is equivalent to such inhomogeneous model, see \cite{GGvdHP, DomGiaGibHofPri16}. Moreover, since $\beta_{ij}$ is close to factorizing 
into a contribution due to $i$ and to $j$,  one can prove  \cite{GGvdHP, DomGiaGibHofPri16} that:\\
	\be
	\label{eq-rewriteQNexp}
	\expec\biggl[\exp\biggl\{\beta \sum_{(i,j)\in E_{\sN}} \sigma_i \sigma_j + B\sum_{i\in[\sN]} \sigma_i\biggr\}\biggr]
	= \e^{\sN \alpha_n(\beta)} \e^{\frac12 \frac{\sinh(\beta)}{\ell_\sN} \left(\sum_{i} w_i \sigma_i\right)^2 +B\sum_{i\in[\sN]} \sigma_i + o(n)}.
	\ee
This computation shows that, in the large $n$-limit, the annealed measure $\muan_{\sN}$ at inverse temperature $\beta$ is close 
to the Boltzmann-Gibbs 
measure ${\mu}^{\sss {\mathrm{ICW}}}_n$ of the rank-1 inhomogeneous Curie-Weiss model 
at inverse temperature $\tilde{\beta}=\sinh(\beta)$
\be
{\mu}^{\sss {\mathrm{ICW}}}_n(\sigma) = \frac{\exp(H^{{\sss {\mathrm{ICW}}}}_n(\sigma))}{Z^{\sss \mathrm{ICW}}_{\sN}(\tb,B)}
\ee 
with Hamiltonian
\be\label{ICWham}
H^{{\sss {\mathrm{ICW}}}}_n(\sigma)= \frac12 \frac{\tilde{\beta}}{\ell_\sN} \left(\sum_{i} w_i \sigma_i\right)^2 +B\sum_{i\in[\sN]} \sigma_i
\ee
and normalizing partition function
	\be
	\label{z-icw}
	Z^{\sss \mathrm{ICW}}_{\sN}(\tb,B)=\sum_{\sigma \in \Omega_{\sN}}
	\e^{B \sum_{i \in[\sN]} {\sigma_i}}\e^{\frac{1}{2}\frac{ \tb}{\ell_{\sN}}\left(\sum_{i\in[\sN]}w_i \sigma_i\right)^2}.
	\ee
The above analysis can be simply extended to moment generating functions involving (some of) the edge variables $(I_{ij})_{1\leq i<j\leq n}$, as these can be incorporated into the exponential term and the expectation w.r.t.\ them can then again be taken. Of course, in such settings, the connection to the rank-1 inhomogeneous Curie-Weiss model is changed as well, and a large part of our paper deals precisely with the description of such changes, as well as their effects.

\section{LDP for the total spin}
\label{sec-LDP-spin}

\subsection{LDP in the high-temperature regime}
We first prove the LDP in the high-temperature regime for the annealed Ising model using the G\"artner-Ellis theorem.

\begin{proof}[Proof of Theorem~\ref{thm-LDP-HT}]
To apply the G\"artner-Ellis theorem we need the thermodynamic limit of the cumulant generating function of $S_{\sN}$ w.r.t.\ $\muan_{\sN}$, given by
	\be
	{c}(t) = \lim_{\sN\to\infty} \frac{1}{\sN} \log{\mathbb{E}_{\muan_{\sN}} \left[\exp\left(t S_{\sN}\right)\right]}.\nn
	\ee
Observe that
	\begin{align}
	\mathbb{E}_{\muan_{\sN}} \left[\exp\left(t S_{\sN}\right)\right] 
	&=\frac{Z^\an_{\sN}(\beta,B+t)}{Z^\an_{\sN}(\beta,B)}.\nn
	\end{align}
Hence,
	\be
	\label{eq-ldp1}
	{c}(t) = \lim_{\sN\to\infty} \frac{1}{\sN} \log{\frac{Z^\an_{\sN}(\beta,B+t)}{Z^\an_{\sN}(\beta,B)}}= \psian(\beta,B+t) - \psian(\beta,B),\nn
	\ee
where the existence of the limit follows from Theorem~\ref{thm-annealedpressure}.
We know that, for $B\neq0$,
	\be
	\frac{\dint}{\dint B} \psian(\beta,B) = M^\an(\beta,B).\nn
	\ee
For $\beta\leq\betacan$,
	\be
	\lim_{B\searrow0} M^\an(\beta,B) = \lim_{B\nearrow0} M^\an(\beta,B) =0,\nn
	\ee
so that $c(t)$ is differentiable in $t$. Hence, it follows from the G\"artner-Ellis theorem~\cite[Thm.~2.3.6]{DemZei09} that $S_{\sN}$ satisfies an LDP with rate function given by the Legendre transform of $c(t)$ which is given by~\eqref{rate2}. The proof for the random quenched Ising model is analogous.
\end{proof}

Let us now elaborate on the interpretation of the above results. The stationarity condition for~\eqref{rate2} is
	\be
	\label{xeqMbetaBplust}
	x= {M^\an}(\beta, B+t),
	\ee
which defines a function $\check{t}=\check{t}(x;\beta, B)$ such that 
	\be
	\label{final1}
	{{I^\an}}(x)= x\, \check{t}(x;\beta, B) - {{\psian}(\beta,B+ \check{t}(x;\beta, B) )}+ \psian(\beta,B).\nn
	\ee
Given $(\beta,B)$, the total spin per particle will concentrate around its typical value $M^\an(\beta,B)$ coinciding with the magnetization. To observe the atypical value $x$ the field must be changed from $B$ to $B+t$, where $t$ is determined by requiring that $x$ is the magnetization $M^\an(\beta,B+t)$. Note that we have not made use of any specifics about the graph sequence, or whether we are in the annealed or quenched setting. Hence, the above holds for Ising models on any graph sequence, as long as the appropriate thermodynamic limit of the pressure exists.

For $\beta>\betacan$, 
	\be
	m^+ := \lim_{B\searrow0} M^\an(\beta,B) > 0 > \lim_{B\nearrow0} M^\an(\beta,B) = -m^+,\nn
	\ee
and hence $c(t)$ is not differentiable for $t=-B$ and the G\"artner-Ellis theorem can no longer be applied. Since the spontaneous magnetization is not zero, it is not possible to find a $t$ such that~\eqref{xeqMbetaBplust} holds for $-m^+<x<m^+$. Therefore, the Legendre transform~\eqref{rate2} has a flat piece.
By the G\"artner-Ellis theorem, this Legendre transform still gives a lower bound on the rate function, but it is only an upper bound for so-called exposed points of the Legendre transform, i.e., for $x$ outside this flat piece. In fact, we show that the Legendre transform in general does {\em not} give the correct rate function, since the Legendre transform of the pressure is convex and we show that the rate function in the low temperature regime in general is not.

\subsection{LDPs for the total spin and weighted spin}\label{seq-LDP-weighted}
In this section we prove Theorem \ref{thm-ldp-weighted} and \col{then we deduce from it a new proof of
Theorem \ref{thm-annealedpressure} (thus by a method different from that of \cite{GGvdHP})}. 
Following Ellis' approach~\cite{Ellis}, we can compute the annealed pressure $\psian(\beta, B)$ and the large deviation function of  
${\bf Y}_{\sN}(\sigma):=(m_{\sN}(\sigma),\mwn(\sigma)) \equiv (\frac{S_n(\sigma)}{n},\frac{S^{(w)}_n(\sigma)}{n})$ w.r.t.\ the annealed measure $\muan_{\sN}$, starting from the LDP of $(m_{\sN},\mwn)$ w.r.t.\ the product measure 
	\be\label{proddelta}
	P_{\sN}=\bigotimes_{i=1}^{N}\left (\frac 1 2 \delta_{-1}+\frac 1 2 \delta_{+1} \right ).
	\ee
The large deviations of ${\bf Y}_{\sN}=(m_{\sN},\mwn)$ w.r.t.\ $P_{\sN}$ can \col{easily} be obtained by applying the G\"artner-Ellis theorem.

\begin{proof}[Proof of Theorem~\ref{thm-ldp-weighted}]
Let ${\bf t}=(t_1,t_2)$ and compute
	\be
	\E_{P_{\sN}} [\exp (\sN\, {\bf t} \cdot {\bf Y}_{\sN})] = \E_{P_{\sN}} [\exp ( t_1 S_{\sN}+ t_2 S^{\sss(w)}_{\sN})] 
	=  \E_{P_{\sN}} [ \Pi_{i\in [\sN]} \exp ( t_1 + w_i t_2)\sigma_i  ] =  \Pi_{i\in [\sN]}  \cosh (t_1 + w_i t_2),\nn
	\ee
where $\E_{P_{\sN}}$ denotes average w.r.t.\ $P_{\sN}$.
Thus, the cumulant generating function of the vector ${\bf Y}_{\sN}=(m_{\sN},\mwn)$ w.r.t.\ $P_{\sN}$  equals
	\be\label{cumgenfY}
	c_{\sN}({\bf t})= \frac{1}{\sN} \log \E_{P_{\sN}} [\exp \sN({\bf  t} \cdot {\bf Y}_{\sN})] = \frac{1}{\sN}\sum_{i\in [\sN]} \log \cosh  (t_1 + w_i t_2)=\E[\log \cosh  (t_1 + W_{\sN} t_2)],\nn
	\ee
here $\E$ represents the average w.r.t.\ the uniformly chosen vertex $W_{\sN}$. 
Since $|\log \cosh  (t_1 + W_{\sN} t_2)| \leq |t_1 + W_{\sN} t_2| \leq |t_1| + W_{\sN} |t_2|$ it follows from Condition~\ref{cond-weightreg}(b) and the dominated convergence theorem that
	\be
 	c({\bf t}):=\lim_{\sN\to \infty} c_{\sN}({\bf t})= \E [ \log \cosh (t_1 + W t_2) ],\nn
	\ee
with $W$ limiting weight of the graph. 
By the G\"artner-Ellis theorem, we conclude that ${\bf Y}_{\sN}$ has a large deviation principle 
with rate function
	\be
	I(x_1,x_2)= \sup_{(t_1,t_2)} \left ( t_1 x_1 + t_2 x_2 - \E [\log \cosh (t_1 + W t_2)]\right ).\nn
	\ee
We have 
	\be\nn
	I(x_1,x_2)=\left \{
	\begin{array}{ll}
	{t^\star_1} x_1 + {t^\star_2} x_2 - \E [\log \cosh ({t^\star_1} + W {t^\star_2})],&\quad  \mbox{if}\; |x_1| < 1, |x_2|< \E[W],\\
	+\infty, & \quad \mbox{otherwise,}
	\end{array}
	\right.
	\ee
where $t^\star_1=t^\star_1(x_1,x_2)$  and $t^\star_2=t^\star_2(x_1,x_2)$ are given by the stationarity condition
	\be\label{eqstatmq}
	\left \{ 
	\begin{array}{l}
	x_1= \E [ \tanh(t_1 + W t_2)],\\
	x_2 =  \E [ W \tanh(t_1 + W t_2)],
	\end{array}
	\right.
	\ee
for $ |x_1| < 1, |x_2|< \E[W]$.

For any function $f \,:\, \Omega_{\sN} \to \bbR$ we can write
	\be
	\sum_{\sigma\in \Omega_{\sN}} f(\sigma) = 2^{\sN} \int_{\Omega_{\sN}} f(\sigma) \dint P_{\sN}(\sigma).\nn
	\ee
Hence, also using~\eqref{eq-rewriteQNexp},
	\begin{align}
	Z^\an_{\sN}(\beta,B) &= 2^{\sN} \e^{\sN \alpha_{\sN}}  \int_{\Omega_{\sN}} \e^{\frac12 \frac{\sinh(\beta)}{\sN\expec[W_{\sN}]} 
	\left(\sum_{i} w_i \sigma_i\right)^2 +B\sum_{i\in[\sN]} \sigma_i + o(n)}\dint P_{\sN}(\sigma) \nn\\
	&= 2^{\sN} \e^{\sN \alpha_{\sN}} \int_{\Omega_{\sN}}   \e^{\frac{\sN}{2} \frac{\sinh(\beta)}{\expec[W_{\sN}]} (\mwn)^2 +\sN B m_{\sN} + o(n)}\dint P_{\sN}(\sigma)\nn
	\end{align}
and, similarly,
	\be
	\muan_{\sN}(\cdot) =  \frac{2^{\sN} \e^{\sN \alpha_{\sN}}}{Z_{\sN}^\an(\beta,B)}\int_{\Omega_{\sN}} (\cdot) \;  
	\e^{\frac{\sN}{2} \frac{\sinh(\beta)}{\expec[W_{\sN}]} (\mwn)^2 +\sN B m_{\sN} + o(n)}\dint 	P_{\sN}(\sigma) .\nn
	\ee

Then, by applying Varadhan's lemma~\cite[Thm.~II.7.1]{Ell06},
	\be
	\label{pressure}
	\psian(\beta,B)=\lim_{\sN\to\infty}\frac{1}{\sN} Z_{\sN}^\an(\beta,B) =\log(2) + \alpha(\beta)+  \sup_{(x_1,x_2)} \left [ \frac{\sinh(\beta)}{2\E[W]}x_2^{2}+Bx_1 - I(x_1,x_2) \right ]\nn
	\ee
\col{which is equivalent to \eqref{p-annealed-2dim}}, and the rate function of $(m_{\sN},\mwn)$ w.r.t.\  the annealed measure is~\cite[Thm.~II.7.2]{Ell06}
	\be
	I^\an_{\beta,B}(x_1,x_2)=I(x_1,x_2) - \frac{\sinh(\beta)}{2\E[W]}x_2^{2}-  B x_1 - \log(2) - \alpha(\beta)+ \psi^\an(\beta,B).\nn
	\ee
This shows that indeed $(S_{\sN}, S^{\sss(w)}_{\sN})$  satisfies an LDP w.r.t.\ $\muan_{\sN}$ with rate function given by~\eqref{eq-jointLDPrate1}.
By applying the contraction principle, we obtain the rate functions $I^\an_{\beta,B}$ of $m_{\sN}$ and  $J^\an_{\beta,B}$ of $\mwn$ as
	\be
	\label{I-contra}
	\quad I^\an_{\beta,B}(x_1)= \inf_{x_2} I^\an_{\beta,B}(x_1,x_2),
	\qquad J^\an_{\beta,B}(x_2)= \inf_{x_1} I^\an_{\beta,B}(x_1,x_2).
	\ee

In a similar way, we can also immediately obtain an LDP by incorporating the magnetic field in the a priori measure on the spins. For this, define 
	\be
	P^{\sss (B)}_{\sN}=\bigotimes_{i=1}^{\sN}\left(\frac{\e^{-B}}{\e^B+\e^{-B}} \delta_{-1}+\frac{\e^{B}}{\e^B+\e^{-B}} \delta_{+1}\right).\nn
	\ee
Then
	\be
	\E_{P^{\sss (B)}_{\sN}} [\exp (\sN\, {\bf t} \cdot {\bf Y}_{\sN})] = \E_{P^{\sss (B)}_{\sN}} [ \prod_{i\in [\sN]} \exp ( t_1 + w_i t_2)\sigma_i  ] 
	=  \prod_{i\in [\sN]}  \frac{\cosh (t_1 + w_i t_2)}{\cosh(B)},\nn
	\ee
where $\E_{P^{\sss (B)}_{\sN}}$ denotes average w.r.t.\ $P_{\sN}^{\sss (B)}$.
Hence, the cumulant generating function is given by
	\be
	\col{c_{\sN}^{\sss (B)}}({\bf t}) =  \E[\log \cosh  (B + t_1 + W_{\sN} t_2) ]- \log \cosh B,\nn
	\ee
(with $\E$ the average w.r.t.\ the uniformly chosen vertex $W_{\sN}$) which, as in the previous case, converges to 
	\be
	c\col{^{\sss (B)}}({\bf t}) = \E[ \log \cosh  (B + t_1+W t_2)] - \log \cosh B.\nn
	\ee
We can apply the G\"artner-Ellis theorem to obtain that $(m_{\sN},\mwn)$ satisfies an LDP w.r.t.\ $P^{\sss (B)}_{\sN}$ with rate function
	\be\label{eq-rateIB}
	I^{\sss (B)}(x_1,x_2) = \sup_{t_1,t_2} \left(t_1 x_1+t_2 x_2 - \E[ \log \cosh  (B + t_1 + W t_2)] \right) + \log \cosh B.
	\ee
The stationarity conditions are given by
	\be
	\label{eqstatmq-rep}
	\left \{ 
	\begin{array}{l}
	x_1= \E [ \tanh(B+t_1 + W t_2)],\\
	x_2 =  \E [ W \tanh(B+t_1 + W t_2)].
	\end{array}
	\right.
	\ee

Note that
	\be
	\sum_{\sigma\in \Omega_{\sN}} f(\sigma)\e^{B\sum_{i\in[\sN]}\sigma_i} = (2\cosh B)^{\sN} \int_{\Omega_{\sN}} f(\sigma) \dint P^{(B)}_{\sN}(\sigma).\nn
	\ee
Hence,
	\be
	\label{approx}
	\muan_{\sN}(\cdot) =  \frac{(2\cosh B)^{\sN} \e^{\sN \alpha_{\sN}}}{\col{Z^{\an}_{\sN}(\beta,B)}}\int_{\Omega_{\sN}} (\cdot) \;  
	\e^{\frac{\sN}{2} \frac{\sinh(\beta)}{\expec[W_{\sN}]} (\mwn)^2 + o(n)}\dint P^{(B)}_{\sN}(\sigma)
	\ee
where
	\be
	Z^\an_{\sN}(\beta,B) = (2\cosh B)^{\sN} \e^{\sN \alpha_{\sN}}\int_{\Omega_{\sN}} 
	\e^{\frac{\sN}{2} \frac{\sinh(\beta)}{\expec[W_{\sN}]} (\mwn)^2 + o(n)}\dint P^{(B)}_{\sN}(\sigma).\nn
	\ee
As above, it immediately follows that $(m_{\sN},\mwn)$ satisfies an LDP w.r.t.\ the annealed measure with rate function
	\be
	I^{\an\sss (B)}_{\beta,B}(x_1,x_2) = I^{\sss (B)}(x_1,x_2) -   \frac{\sinh(\beta)}{2\E[W]} x_2^{2} -\log\cosh B-\log 2-\alpha(\beta) + \psian(\beta,B),\nn
	\ee
where the pressure is given by
	\be
	\label{eq-pressure7}
	\psian(\beta,B) = \sup_{x_1,x_2} \left(  \frac{\sinh(\beta)}{2\E[W]}x_2^{2}- I^{\sss (B)}(x_1,x_2) \right) +\log\cosh B+\log 2+\alpha(\beta).
	\ee
This proves that also~\eqref{eq-jointLDPrate2} is a rate function for the LDP of $(S_{\sN}, S^{\sss(w)}_{\sN})$. The uniqueness of the large deviation function ~\cite[Thm.~II.3.2]{Ell06}
implies that~\eqref{eq-jointLDPrate2} and ~\eqref{eq-jointLDPrate1} coincide.
\end{proof}

We can rewrite the pressure in~\eqref{eq-pressure7} to prove Theorem~\ref{thm-annealedpressure}:

\begin{proof}[Proof of Theorem~\ref{thm-annealedpressure}]
Note that~\eqref{eq-pressure7} is equivalent to
	\be
	\psian(\beta,B) = \sup_{x_2} \left(  \frac{\sinh(\beta)}{2\E[W]} x_2^{2}- \inf_{x_1}I^{\sss (B)}(x_1,x_2) \right)+\log\cosh B+\log 2+\alpha(\beta),
	\ee
where it should be noted that, by the contraction principle, $\inf_{x_1}I^{\sss (B)}(x_1,x_2)$ is equal to the rate function $I^{\sss(w)}$ for the LDP of $\mwn$ w.r.t.\ $P^{\sss (B)}_{\sN}$. Setting $t_1=0$ in the above computations, this can be proved to be
	\be
	\label{eq-rateIw}
	I^{\sss(w)}(x) = \sup_{t} \left(t x - \E[ \log \cosh  (B + W t)] \right) + \log \cosh B,
	\ee
so that
	\be
	\label{eq-pressureLDPq}
	\psian(\beta,B) = \sup_{x_2} \left(  \frac{\sinh(\beta)}{2\E[W]} x_2^{2}-I^{\sss(w)}(x_2) \right)+\log\cosh B+\log 2+\alpha(\beta).
	\ee
The supremum in~\eqref{eq-rateIw} is attained for $t$ satisfying
	\be
	x = \E[W \tanh(B+Wt)] =: f(t).\nn
	\ee
Since $f(t)$ is strictly increasing, its inverse $f^{-1}$ is well defined. Hence, 
	\be
	I^{\sss(w)}(x) = f^{-1}(x) x - \E[ \log \cosh  (B + W f^{-1}(x))] + \log \cosh B,\nn
	\ee
and
	\begin{align}
	\frac{d}{dx}\left(\frac{\sinh(\beta)}{2\E[W]}x^{2}- I^{\sss(w)}(x) \right) &=\frac{\sinh(\beta)}{\E[W]}x- f^{-1}(x) - \left(x-\E[W\tanh(B+W f^{-1}(x))]\right) \frac{d}{dx}f^{-1}(x) \nn\\
	&=\frac{\sinh(\beta)}{\E[W]}x- f^{-1}(x).\nn
	\end{align}
Hence, the supremum in~\eqref{eq-pressureLDPq} for $x$ satisfying $ f^{-1}(x)=\frac{\sinh(\beta)}{\E[W]}x$,  or equivalently,
	\be
	\label{eq-fixedpointx}
	x = f\left(\frac{\sinh(\beta)}{\E[W]}x\right) = \E\left[W\tanh\left(B+\frac{\sinh(\beta)}{\E[W]}W x\right)\right].
	\ee
For any solution $x^\star$ of~\eqref{eq-fixedpointx},
	\begin{align}
	F(x^\star) &:=  \frac{\sinh(\beta)}{2\E[W]} x^{\star 2}- I^{\sss(w)}(x^\star) +\log\cosh B+\log 2+\alpha(\beta) \nn\\
	& = - \frac{\sinh(\beta)}{2\E[W]} x^{\star 2}+\E[ \log \cosh  (B + \frac{\sinh(\beta)}{2\E[W]} W x^\star )] +\log 2+\alpha(\beta).\nn
	\end{align}

For $B>0$, $f(t)$ is an increasing, bounded and concave function for $t\geq0$ with $f(0)>0$, and hence there is a unique positive solution $x^+$ to~\eqref{eq-fixedpointx}. For any negative solution to~\eqref{eq-fixedpointx}, $x^-$ say,
	\be
	F(x^-) < F(-x^-)\leq F(x^+),\nn
	\ee
since $x^+$ is the unique positive local maximum. An analogous argument holds for $B<0$. Hence,
	\be
	\psian(\beta,B) = - \frac{\sinh(\beta)}{2\E[W]} x^{\star 2}+\E[ \log \cosh  (B + \frac{\sinh(\beta)}{2\E[W]} W x^\star )] +\log 2+\alpha(\beta),\nn
	\ee
where $x^\star$ is the unique solution to~\eqref{eq-fixedpointx} with the same sign as $B$. The value for $B=0$ follows from Lipschitz continuity.
This is equivalent to the formulation in~\eqref{eq-annpressure} by making a change of variables $z^\star = \sqrt{\frac{\sinh(\beta)}{\E[W]}} x^\star$.
\end{proof}

\section{LDP for the number of edges: proof of Theorem~\ref{thm-LDP-edges}}
\label{sec-LDP-edges}
So far we have considered large deviations of the total spin. We now  consider observables that depend only on the graph and investigate their large deviation properties w.r.t.\ the
annealed Ising measure. Such an analysis sheds light on what graph structures optimize the Ising Hamiltonian.

\subsection{Strategy of the proof}
In this section, we investigate the large deviation properties for the number of edges $|E_{\sN}| = \sum_{i<j} I_{ij} $ under the annealed Ising model on the generalized random graph, where we recall that $(I_{ij})_{1\leq i<j\leq \sN}$ denote the independent Bernoulli indicators of the event that the edge $ij$ is present in the graph, which occurs with probability $p_{ij}$ in \eqref{eq-prob-edge}. We aim to apply the G\"artner-Ellis theorem, for which we need to compute the generating function of $|E_{\sN}| $ w.r.t.\ the annealed measure $\muan_{\sN}$ given by
	\be\label{gf-edges}
	\expec_{\muan_n}\Big[\e^{t |E_{\sN}|}\Big]
	= \frac{\expec\left[\sum_{\sigma} \e^{\sum_{i<j} I_{ij} (t+\beta \sigma_i\sigma_j) + B \sum_{i\in [\sN]} \sigma_i} \right]}{\expec	
	\left[\sum_{\sigma} \e^{\sum_{i<j} I_{ij} (\beta \sigma_i\sigma_j) + B \sum_{i\in [\sN]} \sigma_i} \right]}.
	\ee
For later purposes, we will generalize the above computation and, introducing the variables $t_{ij}$,  instead compute the generating function of the Bernoulli indicators $(I_{ij})_{ij}$ defined for $\bft=(t_{ij})_{ij}\in {\mathbb R}^{\sN(\sN-1)/2}$
\be
\label{m-gen}
	R_{\beta, B,n}({\mathbf t}) := \expec_{\muan_n}\Big[\e^{\sum_{1\leq i<j\leq \sN} t_{ij} I_{ij}} \Big]
	=\frac{\expec\left[\sum_{\sigma} \e^{\sum_{i<j} I_{ij} (t_{ij}+\beta \sigma_i\sigma_j) + B \sum_{i\in [\sN]} \sigma_i} \right]}{\expec	
	\left[\sum_{\sigma} \e^{\sum_{i<j} I_{ij} (\beta \sigma_i\sigma_j) + B \sum_{i\in [\sN]} \sigma_i} \right]}.
\ee
This can be carried out in a similar way as in \cite{GGvdHP}. Let us focus on the numerator in the previous display, which we denote by ${\cal A}_{\sN}({\mathbf t},\beta, B)$, so that 
\be
\label{mom-gen-function} 
R_{\beta, B,n}({\mathbf t}) = \frac{{\cal A}_{\sN}({\mathbf t},\beta, B)}{{\cal A}_{\sN}({\bf 0},\beta, B)}\;.
\ee 
We have 
	\begin{align}\nonumber
	{\cal A}_{\sN}({\mathbf t},\beta, B)\, &   = \,  \sum_{\sigma \in \Omega_{\sN}} \e^{B \sum_{i \in[\sN]} {\sigma_i}} 
	\expec\left[\e^{ \sum_{i<j}{I_{ij}(t_{ij}+\beta \sigma_i \sigma_j})} \right] \\ \nonumber
	&  = \,  \sum_{\sigma \in \Omega_{\sN}} \e^{B \sum_{i \in[\sN]} {\sigma_i}}
	\prod_{i<j} \expec\big[\e^{I_{ij}(t_{ij}+\beta\sigma_i \sigma_j)}\big]\\ 
	&  = \,  \sum_{\sigma \in \Omega_{\sN}}\e^{B \sum_{i \in[\sN]} {\sigma_i}} \prod_{i<j} \left(\e^{t_{ij}+\beta\sigma_i \sigma_j}p_{ij} + \left(1 - p_{ij}\right)\right).\nn
	\end{align}
We rewrite
	\begin{equation} 
	\e^{t_{ij}+\beta\sigma_i \sigma_j}p_{ij} + \left(1 - p_{ij}\right) \, = \, C_{ij}(t_{ij})\e^{\beta_{ij}(t_{ij})\sigma_i \sigma_j},\nn
	\end{equation} 
where $\beta_{ij}(t_{ij})$ and $C_{ij}(t_{ij})$ are  chosen such that
	\begin{equation} \nn
	\e^{t_{ij}-\beta}p_{ij} + \left(1 - p_{ij}\right) \, =\, C_{ij}(t_{ij})\e^{-\beta_{ij}(t_{ij})}
	\qquad
	\text{and}
	\qquad \e^{t_{ij}+\beta}p_{ij} + \left(1 - p_{ij}\right) 
	\, =\, C_{ij}(t_{ij})\e^{\beta_{ij}(t_{ij})}.
	\end{equation}
From the above system, we get
	\begin{equation}
	\label{defcbeta} 
	\beta_{ij}(t_{ij})\, = \, \frac{1}{2} \log \frac{\e^{t_{ij}+\beta}p_{ij} + \left(1 - p_{ij}\right)}{\e^{t_{ij}-\beta}p_{ij} + \left(1 - p_{ij}\right)},
	\qquad\quad
	C_{ij}(t_{ij}) \, = \, \frac{\e^{t_{ij}} p_{ij} \cosh(\beta) + \left(1 - p_{ij}\right)}{\cosh\left(\beta_{ij}(t_{ij})\right)}.
	\end{equation}
By symmetry $\beta_{ij}(t_{ij})=\beta_{ji}(t_{ij})$. Furthermore, defining $t_{ji} = t_{ij}$ for $1\le i < j \le n$ and 
\be
\label{def-betaii}
\beta_{ii}(t_{ii})=\frac{1}{2} \log \frac{\e^{t_{ii}+\beta}p_{ii} + \left(1 - p_{ii}\right)}{\e^{t_{ii}-\beta}p_{ii} + \left(1 - p_{ii}\right)} \qquad \text{with} \qquad p_{ii}=w_i^2/(\ell_{\sN}+w_i^2)
\ee
we  obtain
	\begin{align} 
	\label{A-form-general}
	{\cal A}_{\sN}({\mathbf{t}},\beta,B) 
	& = \, G_{\sN}({\mathbf{t}},\beta) \sum_{\sigma \in \Omega_{\sN}}
	\e^{B \sum_{i \in[\sN]} {\sigma_i}}\e^{\frac{1}{2}\sum_{i,j \in[\sN]}{\beta_{ij}(t_{ij})\sigma_i \sigma_j}},
	\end{align}
where
	\be
	\label{defG}
 	G_{\sN}({\mathbf{t}},\beta)=\prod_{1\leq i<j\leq \sN}C_{ij} (t_{ij})\prod_{1\leq i \leq \sN}e^{-\beta_{ii}(t_{ii})/2}.
	\ee
The equations \eqref{mom-gen-function} and \eqref{A-form-general} give us an explicit formula for the moment generating function of the edge variables $(I_{ij})_{ij}$ in the annealed $\GRGnw$ that will prove useful throughout the remainder of this paper. 

\subsection{Moment generating function for the number of edges}

Since the moment generating function for the number of edges in \eqref{gf-edges}
can be obtained from $R_{\beta, B,n}({\mathbf t})$ in \eqref{m-gen} by choosing $t_{ij}=t$ for all $1\leq i<j\leq \sN$,
we continue by studying the asymptotics of ${\cal A}_{\sN}({\mathbf{t}},\beta,B)$ for such
case,  which we denote as ${\cal A}_{\sN}(t,\beta,B)$. By a Taylor expansion of $x\mapsto \log(1+x)$,
	\begin{align} 
	\label{betaexpanded}
	\beta_{ij} (t)\, & = \, \frac{1}{2}\log\left(1 + p_{ij}(\e^{t+\beta} -1)\right) 
	- \frac{1}{2}\log\left(1 + p_{ij}(\e^{t-\beta} -1)\right)\nn \\
	& = \,\frac{1}{2}p_{ij}(\e^{t+\beta} -1) - \frac{1}{2}p_{ij}(\e^{t-\beta} -1) + O(p_{ij}^2(\e^{t+\beta} -1)^2 )+ O(p_{ij}^2(\e^{t-\beta} -1)^2 )\nn\\
	& = \, \e^t\sinh(\beta)p_{ij}+ O(p_{ij}^2(\e^{t\pm\beta} -1)^2),
	\end{align}
therefore
	\begin{align} \nn
	{\cal A}_{\sN}(t,\beta,B) 
	& = \, G_{\sN}(t,\beta)  \sum_{\sigma \in \Omega_{\sN}}
	\e^{B \sum_{i \in[\sN]} {\sigma_i}}\exp\Big\{\frac{1}{2} \e^{t}\sinh(\beta)\sum_{i,j \in[\sN]}{ p_{ij}\sigma_i \sigma_j}+O(\sum_{i,j\in [\sN]} p_{ij}^2(\e^{t \pm \beta} - 1)^2)\Big\}.
	\label{N1}
	\end{align}
For any fixed $t$, the term $ O( \sum_{i,j\in [\sN]} p_{ij}^2(\e^{t\pm\beta} -1)^2)$ can be controlled by using  $p_{ij}\le \frac{w_i w_j}{\ell_{\sN}}$ and Condition \ref{cond-weightreg}(c), which implies that
	\begin{equation}\nn
	\Big|\sum_{i,j \in[\sN]}p_{ij}^2\Big|\leq  \sum_{i,j \in[\sN]}\Big(\frac{w_i w_j}{\ell_{\sN}}\Big)^2
	\, = \, \left(\frac{\sum_{i\in[\sN]} w_i^2}{\ell_{\sN}}\right)^2 \, = \, o(\sN),
	\end{equation}
and then, 
	\begin{align} \nn
	{\cal A}_{\sN}(t,\beta,B) 
	& = \, G_{\sN}(t,\beta)  \e^{o(\sN)} \sum_{\sigma \in \Omega_{\sN}}
	\e^{B \sum_{i \in[\sN]} {\sigma_i}}\exp\Big\{\frac{1}{2} \e^{t}\sinh(\beta)\sum_{i,j \in[\sN]}{ p_{ij}\sigma_i \sigma_j}\Big\}.
	\end{align}
We can proceed further and write
	\begin{align}  
	{\cal A}_{\sN}(t,\beta,B) \, 
	&= \,G_{\sN}(t,\beta) \e^{o(\sN)} 
	\sum_{\sigma \in \Omega_{\sN}}\e^{B \sum_{i \in[\sN]} {\sigma_i}}\e^{\frac{1}{2}\e^t \sinh(\beta)
	\sum_{i,j \in[\sN]}\frac{w_i w_j}{\ell_{\sN}}\sigma_i \sigma_j}\nonumber\\
	&= \, G_{\sN}(t,\beta)   \e^{o(\sN)}  \sum_{\sigma \in \Omega_{\sN}}
	\e^{B \sum_{i \in[\sN]} {\sigma_i}}\e^{\frac{1}{2}\frac{\e^t\sinh(\beta)}{\ell_{\sN}}\left(\sum_{i\in[\sN]}w_i \sigma_i\right)^2},\label{N2}\nn
	\end{align}
where we have also used that, under Condition \ref{cond-weightreg}(c),
	\eqn{
	\sum_{i\in[\sN]} \frac{w_i^2}{\ell_{\sN}}=o(n),
	\qquad
	\sum_{i,j\in [\sN]} [\frac{w_iw_j}{\ell_{\sN}}-p_{ij}]=\sum_{i,j\in [\sN]} \frac{w_i^2w_j^2}{\ell_{\sN}(\ell_{\sN}+w_iw_j)}=o(n).\nn
	}
Recalling the definition of the partition function of the Inhomogeneous Curie-Weiss model
we can thus rewrite
	\be
	{\cal A}_{\sN}(t,\beta,B)=G_{\sN}(t,\beta)  \e^{o(\sN)}  \,  Z^{\sss \mathrm{ICW}}_{\sN}(\e^t \sinh(\beta),B)\, ,\nn
	\ee
while the denominator  in \eqref{gf-edges} equals
	\be
	{\cal A}_{\sN}(0,\beta,B)=G_{\sN}(0,\beta)  \e^{o(\sN)}  \,  Z^{\sss \mathrm{ICW}}_{\sN}( \sinh(\beta),B).\nn
	\ee
Therefore, the annealed cumulant generating function of the number of the edges is
	\eqan{
	\label{cum-gen-edges}
	\varphi_{\beta,B,\sN}(t)&:= \frac{1}{\sN} \log \expec_{\muan_n}\big[\e^{t |E_{\sN}|}  \big]\nn\\
	&=\frac{1}{\sN}\log  Z^{\sss \mathrm{ICW}}_{\sN}(\e^t \sinh(\beta),B) - \frac{1}{\sN} \log Z^{\sss \mathrm{ICW}}_{\sN}( \sinh(\beta),B)
	+ \frac{1}{\sN}\log \frac{\,G_{\sN}(t,\beta) }{G_{\sN}(0,\beta)} +o(1).
	}
In order to apply the G\"artner-Ellis theorem, we need to compute the limit of $\varphi_{\beta,B,\sN}(t)$. We can deal with the first and second term in the r.h.s.\ of \eqref{cum-gen-edges} by using the results obtained in \cite{GGvdHP}, in which the limit pressure of the Inhomogeneous Curie-Weiss model has been computed. 
Indeed, from  \cite{GGvdHP}	
	\begin{eqnarray} 
	\label{limannealedp2}
	\psi^{\sss \mathrm{ICW}}(\sinh(\beta),B) & := & \lim_{\sN\rightarrow\infty}\, 
	\frac{1}{\sN} \log Z^{\sss \mathrm{ICW}}_{\sN} \left( \sinh(\beta), B \right) \nn \\
	& = & 
	 \log 2 + \Big[\mathbb{E} \log \cosh \Big(\sqrt{\frac{\sinh(\beta)}{\mathbb{E}\left[W\right]}}W z^\star(\beta,B) + B\Big)\Big] - \frac{z^\star(\beta,B)^2}{2}
	\end{eqnarray}
with $z^\star(\beta,B)$ defined in Theorem \ref{thm-annealedpressure}.
Similarly
	\begin{eqnarray} 
	\label{limannealedp22}
	\psi^{\sss \mathrm{ICW}}(\e^t \sinh(\beta),B) & := & \lim_{\sN\rightarrow\infty}\, 
	\frac{1}{\sN} \log \left(Z^{\sss \mathrm{ICW}}_{\sN} \left( \e^t \sinh(\beta), B \right)\right) \nn \\
	& = & 
	 \log 2 + \Big[\mathbb{E} \log \cosh \Big(\sqrt{\frac{\e^t\sinh(\beta)}{\mathbb{E}\left[W\right]}}W z^\star(t,\beta,B) + B\Big)\Big] - \frac{z^\star(t,\beta,B)^2}{2}.
	\end{eqnarray}
with $z^\star(t,\beta,B)$ be the unique fixed point with the same sign as $B$ of the equation
	\be\label{fixpGRG2}
	z= 
	\mathbb{E}\Big[\tanh \Big(\sqrt{\frac{e^t\sinh(\beta)}{\mathbb{E}[W]}}W z + B\Big) 
	\sqrt{\frac{e^t\sinh(\beta)}{\mathbb{E}[W]}}\,W\Big].
	\ee	
Next, we have to deal with the third term in \eqref{cum-gen-edges} which, recalling \eqref{defG} and \eqref{defcbeta}, we write explicitly as
	\begin{align}
 	&\frac{1}{\sN}\log \frac{\,G_{\sN}(t,\beta) }{G_{\sN}(0,\beta)} =
	\frac{1}{\sN} \sum_{i<j} \log \left ( \frac{\e^t p_{ij}\cosh(\beta) +1-p_{ij}}{p_{ij}\cosh(\beta)+1-p_{ij} } \right )
 	+ \frac{1}{\sN} \sum_{i<j} \log \left ( \frac{\cosh(\beta_{ij}(0))}{\cosh(\beta_{ij} (t))} \right)
	+ \frac{1}{\sN} \sum_{i\in [n]} \left(\frac{\beta_{ii}(0) - \beta_{ii}(t)}{2}\right).\label{treeterms}
	\end{align}
We start by computing the first term in the r.h.s.\ of the previous display, then we show that the remaining terms give a vanishing contribution in the limit. We start by recalling that, on the basis of the Weight Regularity Condition \ref{cond-weightreg}(a) and (c), 
$\ell_{\sN}=\sN(\expec{[W]}+o(1))=O(\sN)$ and $\sum_{1\leq i<j\leq \sN}p_{ij}^2=O(\sN^{-1})$. Thus, we write the first term in \eqref{treeterms} as
	\begin{align}
	&\frac{1}{\sN} \sum_{i<j} \log \left ( \frac{\e^t p_{ij}\cosh(\beta) +1-p_{ij}}{p_{ij}\cosh(\beta)+1-p_{ij} } \right )
	=\frac{1}{\sN} \sum_{i<j} \log \left (1+ \frac{(\e^t-1) p_{ij}\cosh(\beta) }{1+p_{ij}(\cosh(\beta)-1)} \right )\nn\\
	&\qquad=\frac{1}{\sN} \sum_{i<j} \log \left ( 1+ (\e^t-1) p_{ij}\cosh(\beta) +O(p_{ij}^2)\right)= (\e^t-1) \cosh(\beta)  \frac{1}{\sN} \sum_{i<j}  p_{ij} +O(\sN^{-1}),\nn
	\label{firstterm}
	\end{align}
where the Taylor expansions of $1/(1+x)$ and $\log(1+x)$ have been used.  
Therefore,
	\be
	\label{intermediate}
	\lim_{\sN\to \infty } \frac{1}{\sN} \sum_{i<j} \log \left ( \frac{\e^t p_{ij}\cosh(\beta) +1-p_{ij}}{p_{ij}\cosh(\beta)+1-p_{ij} } \right ) = \frac 1 2 (\e^t-1) \cosh(\beta) \expec{[W]},
 	\ee
since $\frac{1}{n} \sum_{i<j} p_{ij} \to \frac{1}{2} \expec [W]$.
By \eqref{betaexpanded} and a Taylor expansion
	\be\nn
	 \log \left ( \frac{\cosh(\beta_{ij}(0))}{\cosh(\beta_{ij} (t))} \right )= O(p^2_{ij}).
	\ee
Then, by Weight Regularity Condition \ref{cond-weightreg}(c) and $p_{ij}\leq w_iw_j/\ell_{\sN}$,
	\be
	\label{doublesumbetaij}
	\frac{1}{\sN} \sum_{i<j} \log \left ( \frac{\cosh(\beta_{ij}(0))}{\cosh(\beta_{ij} (t))} \right ) = \frac{1}{n}\sum_{i<j}  O({p^2_{ij}}) = O(\sN^{-1}).
	\ee
Furthermore
\be
\label{single}
\frac{1}{\sN} \sum_{i\in [n]} \left(\frac{\beta_{ii}(0) - \beta_{ii}(t)}{2}\right) = \frac{1}{\sN} 
\sum_{i\in [n]} O(p_{ii}) =  O(\sN^{-1}),
\ee
where the definition of $\beta_{ii}(t)$ in \eqref{def-betaii} has  been used.
Combining \eqref{treeterms} with the estimates in  \eqref{intermediate}, \eqref{doublesumbetaij}, \eqref{single} leads to
	\be
	\label{using}
	\lim_{\sN\to \infty} \frac{1}{\sN}\log \frac{\,G_{\sN}(t,\beta) }{G_{\sN}(0,\beta)} = \frac 1 2 (\e^t-1) \cosh(\beta) \expec{[W]}.
	\ee
Considering the limit $n\to\infty$ in  \eqref{cum-gen-edges} and
using \eqref{using}, \eqref{limannealedp22} and  \eqref{limannealedp2} , finally gives us
	\begin{align}
	\varphi_{\beta, B}(t)&:=\lim_{\sN\to \infty} \varphi_{\beta, B, \sN}(t)
	=\mathbb{E}\Big[\log \cosh \Big(\sqrt{\frac{\e^t\sinh(\beta)}{\mathbb{E}\left[W\right]}}W z^\star(t,\beta,B) + B\Big)\Big]\label{ann_cgf_edges}\nn\\
	&\qquad- \mathbb{E}\Big[\log \cosh \Big(\sqrt{\frac{\sinh(\beta)}{\mathbb{E}\left[W\right]}}W z^\star(\beta,B) + B\Big)\Big] 
	+ \frac 1 2 \left ( z^\star(\beta,B)^2
	- z^\star(t,\beta,B)^2 \right )\nonumber \\
	&\qquad+ \frac 1 2 (\e^t-1) \cosh(\beta) \expec{[W]}.
	\end{align}
	
\subsection{Conclusion of the proof}	
With \eqref{ann_cgf_edges} in hand, we are finally ready to prove Theorem \ref{thm-LDP-edges}. Equation \eqref{ann_cgf_edges} identifies the infinite-volume limit of the cumulant generating function of the number of edges. By the G\"artner-Ellis theorem, this also identifies the rate function as its Legendre transform, provided that $t\mapsto \varphi_{\beta, B}(t)$ is differentiable. 
We compute the derivative of $t\mapsto \varphi_{\beta, B}(t)$ in \eqref{ann_cgf_edges} explicitly as
	\begin{align}\label{der_cgf_edges}
	\frac{d}{dt}\varphi_{\beta, B}(t)
	&=\mathbb{E}\Big[\tanh \Big(\sqrt{\frac{\e^t\sinh(\beta)}{\mathbb{E}\left[W\right]}}W z^\star(t,\beta,B) + B\Big)\sqrt{\frac{\e^t\sinh(\beta)}{\mathbb{E}\left[W\right]}}W\Big]\times \Big(\frac12 z^\star(t,\beta,B)+\frac{d}{dt}z^\star(t,\beta,B)\Big)\nonumber\\
	&\qquad
	- z^\star(t,\beta,B)\frac{d}{dt}z^\star(t,\beta,B) + \frac 1 2 \e^t\cosh(\beta) \expec{[W]}.
	\end{align}
Since $z^\star(t,\beta,B)$ is the fixed point for the ICW with $\tilde{\beta}=\e^t \sinh(\beta)$, which is an analytic function of $t$, it holds that $z^\star(t,\beta,B)$ is analytic in $t$ for $B\neq 0$ and hence $\frac{d}{dt}z^\star(t,\beta,B)$ exists. 
By \eqref{fixpGRG2}, the first expectation equals $z^\star(t,\beta,B)$, so that the two terms containing the factors $\frac{d}{dt}z^\star(t,\beta,B)$ cancel, and
	\begin{align}
	\frac{d}{dt}\varphi_{\beta, B}(t)
	&=\frac12 z^\star(t,\beta,B)^2
	+ \frac 1 2 \e^t\cosh(\beta) \expec{[W]}.
	\label{der_cgf_edges-fin}
	\end{align}
For $B=0$, $\frac{d}{dt}z^\star(t,\beta,B)$ might not exist in the critical point $\e^t \sinh(\beta) = \tilde{\beta}_c$. However, since the specific heat is finite, both the left and right derivative exist. Therefore, the above argument can be repeated for the left and right derivative, which both give the r.h.s.\ of~\eqref{der_cgf_edges-fin}, so that this equation is also true for $B=0$.


This shows that $t\mapsto \varphi_{\beta, B}(t)$ is 
differentiable and it concludes the proof of the main statement in Theorem \ref{thm-LDP-edges}
about the large deviations function for the number of edges in the annealed $\GRGnw$.
Formula \eqref{number-edges-AIM} for the expected number of edges is immediately
 obtained by evaluating \eqref{der_cgf_edges-fin} in $t=0$.
 
Finally, we note that by the LDP derived in the previous section, and the fact that the limiting rate function is strictly convex (this can be seen by noting that both terms on the r.h.s.\ of \eqref{der_cgf_edges-fin} are strictly increasing) the rate function has a unique minimum, which immediately shows that $|E_{\sN}|/\sN$ is concentrated around its mean, which has already been derived in \eqref{lim-number-edges2} as well as in \eqref{der_cgf_edges-fin}.
\qed

\medskip

\begin{remark}[Moment generating function of total degree for $\GRGnw$]
At  zero magnetic field $B=0$ and infinite temperature $\beta=0$, the annealed average of any function of the graph coincides with  the average with respect to the law of the graph.  Then, $\varphi_{0,0,\sN}(t)$ is the cumulant generating function of the number of edges of the $\GRGnw$. In this case, \eqref{ann_cgf_edges} gives
	\be
	\varphi_{0, 0}(t)=\frac 1 2 (\e^t-1) \expec{[W]},
	\ee
because $z^\star(0,0)=0$, which can also be seen by direct computation.
\end{remark}

\section{Degree distribution under annealed measure: proof of Theorem \ref{thm-degrees-AIM}}
\label{sec-degree-distr-AIM}
Given $(D_i)_{i\in [\sN]}$,  the degree sequence of the $\GRGnw$ we want to compute its moment generating function with respect to the annealed measure  $\muan_{\ch{n}}$, i.e.,
	\be
	g_{\beta,B,\sN}(\bfs)=\expec_{\muan_n}\Big[\e^{\sum_{i \in [\sN]} s_i D_i} \Big],\nn
	\ee
for $\bfs=(s_1,s_2,\ldots, s_{\sN}) \in {\mathbb R}^{\sN}$.
Since $D_i=\sum_{j \neq i} I_{ij}$,  where $(I_{ij})_{1\le i < j \le \sN}$ are the independent  Bernoulli variables  with parameters $p_{ij}$ representing the  indicator that the edge $ij$ exists and $I_{ji}=I_{ij}$,  
we can write $ \sum_{i \in [\sN]} s_i D_i= \sum_{i<j} I_{ij}(s_i+s_j)$, then recalling \eqref{m-gen}
we have
	\be	
	g_{\beta,B,\sN}(\bfs)= R_{\sN,\beta,B}(\bft(\bfs))
	\ee
where we define $t_{ij}(\bfs) := s_i +s_j$ for $1\le i < j \le n$. Furthermore, by \eqref{mom-gen-function},
	\be
	\label{deggenf}
	g_{\beta,B,\sN}(\bfs)
	=\frac{{\cal A}_{\sN}(\bft(\bfs),\beta, B)}{{\cal A}_{\sN}({\mathbf{0}},\beta, B)},
	\ee
where we recall that ${\cal A}_{\sN}(\bft,\beta, B)$ was defined in \eqref{A-form-general}. This is the starting point of our analysis. 
In Section \ref{sec-mom-degrees} we simplify the expression for the
moment generating function of the degrees by using the mapping of the
annealed Ising measure to the rank-1 inhomogeneous Curie-Weiss model.
We then investigate the degree of a fixed vertex under the annealed Ising model
in section \ref{sec-degree-fixed} and we consider finitely many degrees
in section \ref{deg-fin-many}.
\subsection{Moment generating function of the degrees}
\label{sec-mom-degrees}
We start by rewriting the generating function of the degree $g_{\beta,B,\sN}(\bfs)$. To this aim,
due to \eqref{deggenf}, we need to rewrite ${\cal A}_{\sN}({\mathbf{t}}(\bfs),\beta, B)$.
This can be done using again the Hubbard-Stratonovich identity.  
Introducing the standard Gaussian variable $Z$, we will show that we can extend the arguments in \cite{GGvdHP} to show that
	\eqn{
	\label{An-t-asy}
	{\cal A}_{\sN}({\mathbf{t}}(\bfs),\beta, B)=G_{\sN}({\bf t}(\bfs),\beta)\, 2^{\sN}\, \e^{-\kappa({\bf t})} \,\E_Z \left [ \exp\big\{\sum_{i=1}^{\sN} \log\cosh \left (a_{\sN}(\beta)  \e^{s_i} w_i Z+B \right )\big\}\right](1+o(1)),
	}
where $a_{\sN}(\beta) =\sqrt{\frac{\sinh(\beta)}{\ell_{\sN}}}$, $\kappa({\bf t})$ is some appropriate constant and $\E_Z$ denotes the expectation w.r.t.\ the Gaussian variable $Z$. 
This boils down to proving convergence of the moment generating function, which requires sharp asymptotics for ${\cal A}_{\sN}({\mathbf{t}}(\bfs),\beta, B)$, while in  \cite{GGvdHP}, it sufficed to study the logarithmic asymptotics.

To see \eqref{An-t-asy}, we define the $\bfs$-dependent rank-1 inhomogeneous Curie-Weiss model measure as
	\eqn{
	\label{ICW-def}\nn
	\muICWt(\sigma)=\frac{1}{\ZICW_{\sN,\bf{s}}(\sinh(\beta),B)}\e^{\frac12 \sum_{i,j} \sinh(\beta) \e^{s_i}\e^{s_j}\frac{w_iw_j}{\ell_{\sN}} \sigma_i\sigma_j +B\sum_{i\in[\sN]} \sigma_i},
	}
with $\ZICW_{\sN,\bf{s}}(\sinh(\beta),B)$ the appropriate partition function. Then, using \eqref{A-form-general}, we can follow \cite[(4.64)]{DomGiaGibHofPri16} to obtain that
	\be
	\label{dinosaur}
	{\cal A}_{\sN}({\mathbf{t}(\bfs)},\beta, B)=G_{\sN}({\mathbf{t}}(\bfs),\beta) \ZICW_{\sN,\bf{s}}(\sinh(\beta),B) \,\expec_{\muICWt}\Big[\e^{F_{\sN}({\mathbf{s}})}\Big],
	\ee
where now
	\eqn{\nn
	F_{\sN}({\bf s})=\frac12 \sum_{i,j} \big[\beta_{ij}(s_i+s_j)-\e^{s_i + s_j}\sinh(\beta)\frac{w_iw_j}{\ell_{\sN}}\big]\sigma_i\sigma_j,
	}
and we have adapted notation from $E_{\sN}$ in \cite[(4.64)]{DomGiaGibHofPri16} to $F_{\sN}$ here to avoid confusion with the total number of edges. 
To further simplify \eqref{dinosaur}, we observe that, following the proof of \cite[Lemma 4.1]{DomGiaGibHofPri16}, one has
	\eqn{
	 \ZICW_{\sN,\bf{s}}(\sinh(\beta),B) = 2^{\sN}\E_Z \left [ \exp\big\{\sum_{i=1}^{\sN} \log\cosh \left (a_{\sN}(\beta)  \e^{s_i} w_i Z+B \right )\big\}\right].\nn
	}
Further,  under Condition~\ref{cond-weightreg}(a)--(c), we can follow the proof of \cite[Lemma 4.7]{DomGiaGibHofPri16} to identify the limit of $\expec_{\muICWt}\Big[\e^{F_{\sN}({\bfs})}\Big]$, as formulated in the next lemma:
\begin{lemma}[Asymptotics correction term]
\label{lem-Fnt-asy}
Define $W_{\sN}({\bf{s}})=w_U\e^{s_U}$, where $U\in[\sN]$ is a uniform vertex. Assume that ${\bfs}$ is such that $W_{\sN}({\bf{s}})\convd W({\bf{s}})$ and $\expec[W_{\sN}({\bf{s}})^2]\to \expec[W({\bf{s}})^2]$. Then, there exists $\kappa(\bf{s})\geq 0$ such that
	\eqn{
	\lim_{\sN\rightarrow \infty} \expec_{\muICWt}\Big[\e^{F_{\sN}({\bfs})}\Big]=\e^{-\kappa(\bf{s})}.\nn
	}
In particular, $\kappa(\bf{s})=\kappa(\bf{0})$ when $\bfs=(s_1,\ldots,s_n)$ only contains finitely many non-zero coordinates.
\end{lemma}

\noindent
{\it Proof of Lemma \ref{lem-Fnt-asy}.} We follow the proof of \cite[Lemma 4.7]{DomGiaGibHofPri16} to obtain that
	\eqn{
	F_{\sN}({\bf s})=-\frac{1}{2} \sinh(\beta)\cosh(\beta) \Big(\sum_{i\in [\sN]} \e^{s_{i}} \sigma_i\frac{w_i^2}{\ell_{\sN}}\Big)^2+o(1).\nn
	}
Due to the negativity of this term, Lemma \ref{lem-Fnt-asy} follows when we prove that, for some $bar{\kappa}({\bfs})$,
	\eqn{
	\label{convYn}
	\sum_{i\in [\sN]} \e^{s_{i}} \sigma_i\frac{w_i^2}{\ell_{\sN}}\convp \bar{\kappa}({\bfs}),
	}
and then Lemma \ref{lem-Fnt-asy} follows with $\kappa({\bfs})=\tfrac{1}{2}(\bar{\kappa}({\bfs}))^2 \sinh(\beta)\cosh(\beta)$. We proceed to prove \eqref{convYn}, which, in turn, is equivalent to proving that as $n\to\infty$
	\eqn{
	\expec_{\muICWt}\Big[\e^{r\sum_{i\in [\sN]} \e^{s_{i}} \sigma_i\frac{w_i^2}{\ell_{\sN}}}\Big]\rightarrow \e^{r\bar{\kappa}({\bf s})}.\nn
	}
Following 	\cite[(4.71)]{GGvdHP} we start by applying again the Hubbard-Stratonovich identity that
gives
	\begin{eqnarray}
	\expec_{\muICWt}\Big[\e^{r\sum_{i\in [\sN]} \e^{s_{i}} \sigma_i\frac{w_i^2}{\ell_{\sN}}}\Big] 
	& = &
	\frac{\sum_{\sigma\in\Omega_n}\mathbb{E}_Z\Big[\exp\Big\{\sum_i\Big(\frac{r}{\ell_n}\e^{s_{i}}w_i^2 + 
	\sqrt{\frac{\sinh(\beta)}{\ell_n}}\e^{s_{i}} w_i Z + B \Big)\sigma_i\Big\}\Big]}
	{\sum_{\sigma\in\Omega_n}\mathbb{E}_Z\Big[\exp\Big\{\sum_i\Big( 
	\sqrt{\frac{\sinh(\beta)}{\ell_n}}\e^{s_{i}} w_i Z + B \Big)\sigma_i\Big\}\Big]}.\nn 
	\end{eqnarray}
The sum over the spins can now be performed yielding
	\begin{eqnarray}
	\expec_{\muICWt}\Big[\e^{r\sum_{i\in [\sN]} \e^{s_{i}} \sigma_i\frac{w_i^2}{\ell_{\sN}}}\Big] 
	& = &
	\frac{\mathbb{E}_Z\Big[\exp\Big\{\sum_i \log\cosh\Big(\frac{r}{\ell_n}\e^{s_{i}}w_i^2 + 
	\sqrt{\frac{\sinh(\beta)}{\ell_n}}\e^{s_{i}} w_i Z + B \Big)\Big\}\Big]}
	{\mathbb{E}_Z\Big[\exp\Big\{\sum_i \log\cosh\Big( 
	\sqrt{\frac{\sinh(\beta)}{\ell_n}}\e^{s_{i}} w_i Z + B \Big)\Big\}\Big]}.\nn
	\end{eqnarray}
By introducing the random variables $W_{\sN}({\bf{s}})=w_U\e^{s_U}$, where $U\in[\sN]$ is a uniform vertex,
the previous expression can be rewritten as
	\begin{eqnarray}
	\expec_{\muICWt}\Big[\e^{r\sum_{i\in [\sN]} \e^{s_{i}} \sigma_i\frac{w_i^2}{\ell_{\sN}}}\Big] 
	& = &
	\frac
	{\int_{\mathbb{R}} \exp\Big\{-z^2/2 + n \mathbb{E}\Big[\log\cosh\Big(\frac{r}{\ell_n} W_n^2(\bfs/2)+ \sqrt{\frac{\sinh(\beta)}{\ell_n}}W_n(\bfs) z + B \Big)\Big]\Big\} dz}
	{\int_{\mathbb{R}}  \exp\Big\{-z^2/2 + n \mathbb{E}\Big[\log\cosh\Big(\sqrt{\frac{\sinh(\beta)}{\ell_n}}W_n(\bfs) z + B \Big)\Big]\Big\} dz}.\nn
	\end{eqnarray}
We do a change of variables replacing $\frac{z}{\sqrt{n}}$ by $z$, so that
	\begin{eqnarray}
	\expec_{\muICWt}\Big[\e^{r\sum_{i\in [\sN]} \e^{s_{i}} \sigma_i\frac{w_i^2}{\ell_{\sN}}}\Big] 
	& = &
	\frac
	{\int_{\mathbb{R}} \exp\Big\{-n z^2/2 + n \mathbb{E}\Big[\log\cosh\Big(\frac{r}{\ell_n} W_n^2(\bfs/2)+ \sqrt{\frac{\sinh(\beta)}{\mathbb{E}[W_n]}}W_n(\bfs) z + B \Big)\Big]\Big\} dz}
	{\int_{\mathbb{R}}  \exp\Big\{-n z^2/2 + n \mathbb{E}\Big[\log\cosh\Big(\sqrt{\frac{\sinh(\beta)}{\mathbb{E}[W_n]}}W_n(\bfs) z + B \Big)\Big]\Big\} dz}.\nn
	\end{eqnarray}
Assuming that $W_{\sN}({\bf{s}})\convd W({\bf{s}})$ for some limiting distribution, as well as 
$\expec[W_{\sN}({\bf{s}})^2]\to\expec[W({\bf{s}})^2]$ (which in fact is a condition on ${\bf{s}}$), 
an application of the  Laplace method yields
	\eqn{
	\expec_{\muICWt}\Big[\e^{r\sum_{i\in [\sN]} \e^{s_{i}} \sigma_i\frac{w_i^2}{\ell_{\sN}}}\Big]
	= \exp\left[{r\expec\Big[\tanh\Big(\sqrt{\frac{\sinh(\beta)}{\expec[W]}}W({\bf{s}})z^\star(\bfs, \beta,B) + B\Big)\frac{W(\frac{{\bf{s}}}{2})^2}{\expec[W]}\Big]}\right](1+o(1))\nn
	}
where $z^\star(\bfs, \beta,B)$ is the solution with the same sign as $B$ of
	$$
	z= 
	\mathbb{E}\Big[\tanh \Big(\sqrt{\frac{\sinh(\beta)}{\mathbb{E}[W]}}W(\bfs) z + B\Big) 
	\sqrt{\frac{\sinh(\beta)}{\mathbb{E}[W]}}\,W(\bfs)\Big].
	$$
All in all, the previous computation shows that \eqref{convYn} holds with
 	\eqn{\nn
	\bar{\kappa}({\bf s})=\expec\Big[\tanh\Big(\sqrt{\frac{\sinh(\beta)}{\expec[W]}}W({\bf{s}})z^\star(\bfs, \beta,B)\Big)\frac{W(\frac{\bf{s}}{2})^2}{\expec[W]}\Big].
	}
When ${\bf{s}}$ only has a finite number of non-zero coordinates, it holds that $W_{\sN}({\bf{s}})\convd W$ and $\expec[W_{\sN}({\bf{s}})^2]\to\expec[W^2]$, so that $\bar{\kappa}({\bf s})=\bar{\kappa}({\bf 0})$, as required.
\qed




\medskip
Armed with \eqref{An-t-asy}, we recall \eqref{deggenf} and thus conclude that
the moment generating function of the degrees is given by
	\be
	\label{gendHS}
	g_{\beta,B,\sN}(\bfs)= (1+o(1))\e^{\kappa(\bf{0})-\kappa(\bf{s})}
	\frac{G_{\sN}(\bft(\bfs),\beta)\E_Z \left [ \exp \sum_{i=1}^{\sN} \log\cosh \left (a_{\sN}(\beta) \e^{s_i} w_i Z+B \right ) \right ] } 
	{G_{\sN}({\bf 0},\beta)\E_Z \left [ \exp \sum_{i=1}^{\sN} \log\cosh \left (a_{\sN}(\beta)  w_i Z+B \right ) \right ]  },
	\ee
with
	\be
	a_{\sN}(\beta) =\sqrt{\frac{\sinh(\beta)}{\ell_{\sN}}}=O\left (\sN^{-\frac 1 2 }\right).\nn
	\ee

\subsection{Degree of a fixed vertex: proof of Theorem \ref{thm-degrees-AIM}}
\label{sec-degree-fixed}

We want to study the distribution of the degree of a fixed vertex. With no loss of generality we can fix, for instance, vertex $i=1$. Thus, we choose $\bfs=\bfs_1$ with $\bfs_1=(s,0,\ldots,0)$, and write
	\be
 	\exp \left [ \sum_{i=1}^{\sN} \log\cosh \left (a_{\sN}(\beta) \e^{s_i} w_i Z+B \right ) \right ]  =\frac{\cosh(a_{\sN}(\beta) \e^s w_1 Z+B)}
	 {\cosh(a_{\sN}(\beta)  w_1 Z+B)}  \exp \left [ \sum_{i=1}^{\sN} \log\cosh \left (a_{\sN}(\beta) w_i Z+B \right ) \right ].\nn
	\ee
Defining 
	\be\label{defh}
	h_{\sN}(Z;\beta,B)
	:= \exp \left \{ \sum_{i=1}^{\sN} \log\cosh \left (a_{\sN}(\beta) w_i Z+B \right ) \right \} 
	= \exp \left \{  \sN \E_{W_n} \left [ \log\cosh \left (a_{\sN}(\beta) W_{\sN} Z+B \right ) \right ] \right \},
	\ee
where $\E_{W_n}$ is the average w.r.t.\ $W_{\sN}=w_U$ being $U$ an uniformly chosen vertex in $[\sN]$, we can introduce the probability measure on $\R$ by
	\be
	\gamma_{\beta,B,\sN}(\cdot):=  \frac{ \E_Z [\;  \cdot  \;  h_{\sN}(Z;\beta,B) ] }{  \E_Z [ h_{\sN}(Z;\beta,B) ]   },\nn
	\ee
and write \eqref{gendHS}  as
	\be\label{genf1}
	g_{\beta,B,\sN}(\bfs_1)=(1+o(1))
	\frac  { G_{\sN}(\bft(\bfs_1),\beta)}{G_{\sN}({\bf 0},\beta)}
	\mathbb{E}_{\gamma_{\beta,B,\sN}}\left (  \frac{\cosh \left (a_{\sN}(\beta) \e^s w_1 Z+B \right )} 
	{\cosh  \left (a_{\sN}(\beta)  w_1 Z+B \right )}   \right),
	\ee
since, by Lemma \ref{lem-Fnt-asy}, $\kappa(\bf{t})=\kappa(\bf{0})$.\\
Now, under the measure $\gamma_{\beta,B,\sN}$, $Z/\sqrt{\sN}\convp z^\star(\beta,B)$, which can be seen by performing a Laplace method on the integral
	\eqn{
	\E_Z [\;  \cdot  \;  h_{\sN}(Z;\beta,B) ] =\int_{-\infty}^{+\infty} \;  \cdot  \; \exp \left [ \sum_{i=1}^{\sN} \log\cosh \left (a_{\sN}(\beta) w_i Z+B \right ) \right ] \e^{-z^2/2}\frac{dz}{\sqrt{2\pi}}.\nn
	} 
In fact, that is precisely the interpretation that $z^\star(\beta,B)$ in Theorem \ref{thm-annealedpressure} has. As a result, 
	\be
	\mathbb{E}_{\gamma_{\beta,B,\sN}}\left (  \frac{\cosh \left (a_{\sN}(\beta) \e^s w_1 Z+B \right )} 
	{\cosh \left (a_{\sN}(\beta)  w_1 Z+B \right )}   \right)\rightarrow \frac{\cosh\Big(z^\star(\beta,B) \e^s w_1\sqrt{\frac{\sinh(\beta)}{\expec[W]}} +B \Big)} 
	{\cosh\Big(z^\star(\beta,B) w_1\sqrt{\frac{\sinh(\beta)}{\expec[W]}} +B \Big)}.\nn
	\ee

Thus,
	\be
	\label{hi}
	\expec_{\muan_n}\big[\e^{s D_1}\big]= (1+o(1)) \frac{G_{\sN}(\bft(\bfs_1),\beta)}{G_{\sN}({\bf 0},\beta)} 
	\frac{\cosh \Big (z^\star(\beta,B)\e^s w_1 \sqrt{\frac{\sinh(\beta)}{\expec[W]}} +B \Big)} 
	{\cosh \Big (z^\star(\beta,B)w_1\sqrt{\frac{\sinh(\beta)}{\expec[W]}} +B \Big)}
	\ee
and we are left with the problem of studying the limit of $G_{\sN}(\bft(\bfs_1),\beta)/G_{\sN}({\bf 0},\beta)$. We have
	\be\label{ratioGGexplicit}
 	\frac  { G_{\sN}(\bft(\bfs_1),\beta) }{G_{\sN}({\bf 0},\beta)\, } = \frac{\prod_{j>1} C_{1j}(s)\, 
	\e^{-\beta_{11}(2s)/2}\, \cdot \, \prod_{1<i<j} C_{ij}(0)\, \prod_{i>1}  \e^{-\beta_{ii}(0)/2}     }
	{\prod_{j>1} C_{1j}(0)\, \e^{-\beta_{11}(0)/2}\, \cdot \, \prod_{1<i<j} C_{ij}(0)\, \prod_{i>1}  \e^{-\beta_{ii}(0)/2} } = \prod_{j>1} \left (\frac{C_{1j}(s)}{C_{1j}(0)} \right)
	\cdot \frac{ \e^{-\beta_{11}(2s)/2}}{\e^{-\beta_{11}(0)/2}},
	\ee
where \eqref{defG} has been used.  From the definition of $C_{ij}(s)$'s, we get
 	\be
	\label{prodCC}
  	\prod_{j>1} \left (\frac{C_{1j}(s)}{C_{1j}(0)} \right) = \prod_{j>1} \frac{\e^{s} \cosh(\beta) p_{1j}+1-p_{1j} }{ \cosh(\beta) p_{1j}+1-p_{1j}}
	\cdot \prod_{j>1} \frac{\cosh(\beta_{1j}(0) )}{\cosh (\beta_{1j}(s))}.
 	\ee
Putting  $p_{ij}=w_i w_j/(\ell_{\sN}+w_i w_j)$, the first term in the l.h.s.\ is rewritten as
 	\be\label{symphfantastique}
  	\prod_{j>1} \frac{\e^s \cosh(\beta) p_{1j}+1-p_{1j} }{ \cosh(\beta) p_{1j}+1-p_{1j}}= \prod_{j>1} \frac{\ell_{\sN} + \e^s \cosh(\beta)w_1 w_j }{\ell_{\sN} + \cosh(\beta) w_1 w_j}  = \e^{\cosh(\beta) w_1 (\e^s-1)}(1+o(1)) \nn
 	\ee
as $\sN\to \infty$. Next, we consider the second factor in the r.h.s.\ of \eqref{prodCC}.  Arguing as in the previous section  for  equation \eqref{doublesumbetaij}, 
 	\be
	\label{berlioz}
  	\sum_{1<j} \log \left ( \frac{\cosh(\beta_{ij}(0))}{\cosh(\beta_{ij} (s))} \right ) = \sum_{1<j}  O(p_{1j}^2)\leq w_1^2\sum_{1<j}  \frac{w_j^2}{\ell_{\sN}^2}=o(1),
 	\ee
since $\max_{j\in[n]} w_j=o(\sN)$. Taking the exponential of the previous relation, we obtain
 	\be
 	\prod_{j>1} \frac{\cosh(\beta_{1j}(0) )}{\cosh (\beta_{1j}(s))} =1+o(1),\nn
 	\ee
as $\sN\to \infty$. Finally, since $\beta_{ij}(s)=o(1)$ as $\sN\to \infty$ (since $p_{ij}\to 0$ in the same limit),  the second factor in the r.h.s.\ of \eqref{ratioGGexplicit} is $1+o(1)$. This proves that
 	\be
  	\frac  { G_{\sN}(\bft(\bfs_1),\beta) }{G_{\sN}({\bf 0},\beta)\, } = \e^{\cosh(\beta) w_1 (\e^s-1)}(1+o(1)).\nn
 	\ee 
 and from  \eqref{hi}, we finally obtain
         \be
 	\expec_{\muan_n}\big[\e^{sD_1}\big]=  (1+o(1)) \e^{\cosh(\beta) w_1 (\e^s-1)}\frac{\cosh\Big(z^\star(\beta,B) \e^s w_1 \sqrt{\frac{\sinh(\beta)}{\expec[W]}} +B \Big)} 
	{\cosh\Big(z^\star(\beta,B) w_1 \sqrt{\frac{\sinh(\beta)}{\expec[W]}} +B \Big)},\nn
 	\ee
as required. \qed

\medskip

\subsection{Degree of a fixed number of vertices: proof of Theorem \ref{thm-degrees-m-AIM}}
\label{deg-fin-many}
We can generalize the previous computation by considering the degrees $(D_1,D_2,\ldots, D_m)$, with $m\in [\sN]$ fixed. The generating function of this random vector 
can be obtained by plugging  $\bfs = \bfs_m$ with ${\bfs}_m=(s_1,s_2,\ldots,s_m,0,\ldots,0)$
into  \eqref{gendHS}. 
By the same arguments of the previous section, we obtain
	\be\label{genf1m}
	g_{\beta,B,\sN}({\bfs}_m)=(1+o(1))
	\frac  { G_{\sN}({\bft}(\bfs_m),\beta)}{G_{\sN}({\bf 0},\beta)}
	\mathbb{E}_{\gamma_{\beta,B,\sN}}\left ( \prod_{i=1}^m  \frac{\cosh \left (a_{\sN}(\beta) \e^{s_i} w_i Z+B \right )} 
	{\cosh  \left (a_{\sN}(\beta)  w_i Z+B \right )}   \right),
	\ee
with 
	\be
	\mathbb{E}_{\gamma_{\beta,B,\sN}}\left ( \prod_{i=1}^m  \frac{\cosh \left (a_{\sN}(\beta) \e^{s_i} w_i Z+B \right )} 
	{\cosh  \left (a_{\sN}(\beta)  w_i Z+B \right )}   \right) \rightarrow \prod_{i=1}^m\frac{\cosh\Big(z^\star(\beta,B) \e^{s_i} w_i\sqrt{\frac{\sinh(\beta)}{\expec[W]}} +B \Big)} 
	{\cosh\Big(z^\star(\beta,B) w_i\sqrt{\frac{\sinh(\beta)}{\expec[W]}} +B \Big)}
	\ee
as $n\to \infty$.
Now we have to study  the limit of $ G_{\sN}({\bft}(\bfs_m),\beta)/ G_{\sN}({\bf 0}_m,\beta)$. From the definition of $G_{\sN}({\mathbf{t}},\beta)$  given in  \eqref{defG} and recalling that $t_{ij}(\bfs)=s_i+s_j$, 
	 \be\label{GoverG3}
	\frac  { G_{\sN}({\bft}(\bfs_m),\beta) }{G_{\sN}({\bf 0},\beta)} = \prod_{ 1\le i < j \le m} \left ( \frac{C_{ij}(s_i+s_j)} {C_{ij}(0)} \right )\cdot \prod_{ 1\le i  \le m \atop j >m } \left (\frac{C_{ij}	
	(s_{i})}{ C_{ij}(0) } \right) \cdot \prod_{i=1}^m \left ( \frac{\e^{-\beta_{ii}(2s_i)/2}}{\e^{-\beta_{ii}(0)/2}} \right).
	\ee
We analyze the three factors separately:\\ \\
$\bullet$ {\em First  and third factors of \eqref{GoverG3}}. By the definition of $C_{ij} (t_{ij})$, 
	\be\label{CoverCfactor1}
	\prod_{ 1\le i < j \le m} \left ( \frac{C_{ij}(s_i+s_j)} {C_{ij}(0)} \right ) 
	= \prod_{ 1\le i < j \le m} \frac{\e^{s_i} \e^{s_j} \cosh(\beta) p_{ij}+1-p_{ij} }{ \cosh(\beta) p_{ij}+1-p_{ij}}\cdot \prod_{ 1\le i < j \le m} \frac{\cosh(\beta_{ij}(0) )}{\cosh (\beta_{ij}(s_i+s_j))},
	\ee
where, by definition of $p_{ij}$,
	\be
 	\prod_{ 1\le i < j \le m} \frac{\e^{s_i} \e^{s_j} \cosh(\beta) p_{ij}+1-p_{ij} }{ \cosh(\beta) p_{ij}+1-p_{ij}}
	= \prod_{1\le i < j \le m} \frac{\ell_{\sN} + \e^{s_i} \e^{s_j} \cosh(\beta)w_i w_j }{\ell_{\sN} + \cosh(\beta) w_i w_j}.\nn
	\ee
We show that this factor is $1+o(1)$. Indeed, following \cite{vdH}, we expand \col{$\log (1+x)$} obtaining:
	\begin{eqnarray}\label{fidelio}
	\log \prod_{1\le i < j \le m} \frac{\ell_{\sN} + \e^{s_i} \e^{s_j} \cosh(\beta)w_i w_j }{\ell_{\sN} + \cosh(\beta) w_i w_j} 
	& = & \frac{\cosh(\beta)}{\ell_{\sN}}  \sum_{1\le i < j \le m} w_i w_j (\e^{s_i} \e^{s_j} -1) \nn \\
	&&\quad\frac{\cosh(\beta)}{\ell_{\sN}^2} \, O(\sum_{1\le i < j \le m} w^2_i w^2_j ) \nn \\
	& = &O(\sN^{-1}) \nn,
	\end{eqnarray}
since $\ell_{\sN}=O(\sN)$ and $m$ is fixed. The second term in the r.h.s.\ of \eqref{CoverCfactor1} and the third factor of  \eqref{GoverG3} converge to 1.
Thus we have shown that  that the first  and third factors of \eqref{GoverG3}  are $1+o(1)$.\\ \\
$\bullet$ {\em Second factor of \eqref{GoverG3}}.  For any fixed $1\le i \le m$,
	\begin{align}
	\label{leonore}
 	\prod_{ j >m } \left (\frac{C_{ij}(s_{i})}{ C_{ij}(0) } \right) = \prod_{j > m} \frac{\ell_{\sN} + \e^{s_i} \cosh(\beta)w_i w_j }{\ell_{\sN} + \cosh(\beta) w_i w_j} 
	\cdot \prod_{j > m} \frac{\cosh(\beta_{ij}(0) )}{\cosh (\beta_{ij}(s_i))}.\nn
	\end{align}

The second factor in the r.h.s.\ of the previous display can be treated as in \eqref{berlioz}, showing that it is $1+o(1)$, while the first factor is close to the generating function of $D_i$ in a GRG  with vertex set $\{i, m+1,\ldots,  \sN\}$ and
weight of vertex $i$ given by $\cosh(\beta) w_i$. We can deal with this term as we have already done, that is,
	\be
	\log  \prod_{j > m} \frac{\ell_{\sN} + \e^{s_i} \cosh(\beta)w_i w_j }{\ell_{\sN} + \cosh(\beta) w_i w_j} =\cosh(\beta) w_i  (\e^{s_i}-1)  \frac{1}{\ell_{\sN}}  \sum_{j > m} w_j + \frac{\cosh(\beta)}	{\ell^2_{\sN}}\,  O( \sum_{j>m} w_j^2).\nn
	\ee
Since $m$ is fixed $ \frac{1}{\ell_{\sN}}  \sum_{j > m} w_j =1+o(1)$, and $\frac{1}{\ell^2_{\sN}}\,  O( \sum_{j>m} w_j^2)=o(1)$, for sufficiently fast decay of $w_i$'s.
Then,
	\be
	 \prod_{j > m} \frac{\ell_{\sN} + \e^{s_i} \cosh(\beta)w_i w_j }{\ell_{\sN} + \cosh(\beta) w_i w_j} = \e^{ \cosh(\beta) w_i  (\e^{s_i}-1)  }(1+o(1)),\nn
	\ee
and the second factor in \eqref{GoverG3} is 
$
\prod_{i=1}^m  \e^{ \cosh(\beta) w_i  (\e^{s_i}-1)  }(1+o(1)).\nn
$
Thus we conclude that
	\be
	\frac  { G_{\sN}({\bft}(\bfs_m),\beta) }{G_{\sN}({\bf 0},\beta)} = \prod_{i=1}^m  \e^{ \cosh(\beta) w_i  (\e^{s_i}-1)  }(1+o(1)).\nn
	\ee
Going back to \eqref{genf1m}, we finally obtain that
	$$
	\expec_{\muan_n}\Big[\e^{\sum_{i=1}^m s_i D_i} \Big] = \prod_{i=1}^m  \e^{ \cosh(\beta) w_i  (\e^{s_i}-1)  }
	\prod_{i=1}^m\frac{\cosh\Big(z^\star(\beta,B) \e^{s_i} w_i\sqrt{\frac{\sinh(\beta)}{\expec[W]}} +B \Big)} 
	{\cosh\Big(z^\star(\beta,B) w_i\sqrt{\frac{\sinh(\beta)}{\expec[W]}} +B \Big)} (1+o(1)),
	$$
as required. \qed

\appendix
\section{Appendix: \col{LDP for the total spin using combinatorial arguments}}
\label{sect_como}
\col{In this appendix, we obtain the large deviation function of the total spin in the
rank-1 inhomogeneous Curie-Weiss model (and thus in the annealed Ising model)
by employing direct combinatorial arguments.
We will restrict to 
the {\em finite-type setting} in which, roughly, there is 
a finite set of values for $w_i$'s.} More precisely,  we define this setting as follows: 
\begin{cond}[Finite-type setting] The vertex weight sequences $\boldsymbol{w} = (w_i)_{i \in [\sN]}$ satisfy the following conditions:
\label{cond-weightfinitetype}
\begin{enumerate}[(a)]
\item There exists a $K\in\mathbb{N}$ and a set of positive numbers $ {\tt A}=\{a_1,a_2,\ldots, a_K\}$, with $a_1<a_2<\ldots < a_K$,  such that  $w_i\in {\tt A}$ for all $i\in[n]$;
\item Denoting by $\hat{\sN}_k (n)$  the number of weights $(w_i)_{i \in [\sN]}$ such that $w_i = a_k$, then  the following limits exist 
	\be
	\lim_{\sN\to\infty} \frac{\hat{\sN}_k (n)}{\sN} = p_k\, ,\quad  k=1,\ldots,K,\nn
	\ee
(obviously $p=(p_1,\ldots,p_K)$ is a probability vector). We define also $\hat{p}_k(n):=\frac{\hat{\sN}_k (n)}{\sN}$ and $e_k(n):=\hat{p}_k(n)-p_k $.
\end{enumerate}
Hereafter, for the sake of notation we drop $n$ from the notation of  $\hat{\sN}_k (n)$, $\hat{p}_k(n)$, $e_k(n)$.
\end{cond}
In this finite-type setting, the previous Condition \ref{cond-weightfinitetype}  is equivalent to Condition \ref{cond-weightreg} in which $W_n$ is the  uniformly chosen weight with
\be\nn
\expec [W_n]= \sum_{k=1}^K a_k \hat{p}_k\, ,\quad \expec [W^2_n]= \sum_{k=1}^K a^2_k \hat{p}_k\, ,
\ee
and $W$ is the limit weight assuming values $a_k$  with probability $p_k$, so that
	\be
	\expec [W] =\sum_{k=1}^K a_k p_k\, , \quad \expec [W^2] =\sum_{k=1}^K a^2_k p_k\, .
	\ee
Assuming Condition \ref{cond-weightfinitetype},  we consider the Hamiltonian \eqref{ICWham}
and  defining  
	\be
	m_{\sN}=\frac{1}{\sN} \sum_{i\in [\sN]} \sigma_{i},\qquad\quad m^{\sss(w)}_{\sN}=\frac{1}{\sN} \sum_{i\in [\sN]}  w_{i}\sigma_{i},
	\ee
we rewrite
	\be\label{hamICWmwmn}
	H^{{\sss {\mathrm{ICW}}}}_{\sN}(\sigma)=\tb \frac{\sN}{2\expec [W]} ({\mnw})^2 + \sN B\ m_{\sN}\, .
	\ee
In the theorem below, we write $\lfloor x \rfloor$ for the integer part of $x>0$.
\begin{theorem}[LDPs for the total spin in the finite-type ICW model] 
\label{thm-LDP-ICW}
In the inhomogeneous Curie-Weiss model defined by  \eqref{ICWham}, and assuming the  finite-type setting in Condition \ref{cond-weightfinitetype}, the total spin $S_n$   
satisfies that for $m\in (-1,1)$, with ${\cal A}= \left [ \frac{1}{2}(1+m)\, a_1  ,  \frac{1}{2}(1+m)\, a_K  \right ]$,
	\begin{align}\label{ldevcomb}
	\lim_{\sN\to\infty} \frac{1}{\sN} \log \mathbb{P}_{{\mu}^{\sss {\mathrm{ICW}}}_n}(S_{\sN}= \lfloor  m\, n \rfloor )= & 
	- \inf_{x\in {\cal A}} \left [  -\frac{\tb }{2} \expec [W] - \frac{2\tb}{\expec [W]} x^2 + 2 \tb x - 
	B\, m + {\tI_m}(x) +\psi^{\sss {\mathrm{ICW}}} (\tb, B) \right] ,
	\end{align}
where $\psi^{\sss {\mathrm{ICW}}} (\tb, B)$ is the pressure of  the model
	and where
	\begin{align}\label{ratefIcomb}
	\tI_m(x)&=  \expec \left [  \frac{\e^{\lambda_1 W + \lambda_2}}{1+\e^{\lambda_1 W + \lambda_2} } \log \left ( \frac{\e^{\lambda_1 W + \lambda_2}}{1+\e^{\lambda_1 W + \lambda_2} } \right ) 
	+ \frac{1}{1+\e^{\lambda_1 W + \lambda_2} } \log \left ( \frac{1}{1+\e^{\lambda_1 W + \lambda_2} } \right )\right ],
	\end{align}		
with $\lambda_1=\lambda_1(x,m)$,  $\lambda_2=\lambda_2(x,m)$ defined implicitly by 
	\be
	\label{eqlambda12}
	\left \{
	\begin{array}{l}
	\expec \left [   \frac{\displaystyle \e^{\lambda_1 W +\lambda_2}}{\displaystyle 1+ \e^{\lambda_1 W +\lambda_2}}    \right ]  = \frac{\displaystyle 1+m}{\displaystyle 2},\\ \\
 	\expec \left [  W   \frac{\displaystyle \e^{\lambda_1 W +\lambda_2}}{ \displaystyle 1+ \e^{\lambda_1 W +\lambda_2}} \right ] = {\displaystyle x}.
	\end{array}
	\right.
	\ee
\end{theorem}
\medskip
\begin{remark} 
\col{The expression for the large deviation rate function of the total spin in the Theorem \ref{thm-LDP-ICW} coincides with the one that is obtained from Theorem \ref{thm-ldp-weighted} by application of the contraction principle and the relation between the annealed Ising model and the inhomogeneous Curie-Weiss model. Indeed, recalling that the annealed measure $\muan_{\sN}$ at inverse temperature $\beta$ is close 
to the Boltzmann-Gibbs measure ${\mu}^{\sss {\mathrm{ICW}}}_n$ of the  inhomogeneous Curie-Weiss model at inverse temperature $\tilde{\beta}=\sinh(\beta)$ (in the sense of equation \eqref{approx})
and by using  $\psi^{\sss {\mathrm{ICW}}}(\tb,B) =  - \alpha(\beta)+ \psi^\an(\beta,B)$,
one finds that the large deviation function of the total spin in the inhomogeneous Curie-Weiss model obtained from \eqref{I-contra} reads
\be
\label{pippo}
I(m) = \inf_{x_2} \left[I(m,x_2) - \frac{\tb}{2\E[W]}x_2^{2}-  B m - \log(2) + \psi^{\sss {\mathrm{ICW}}}(\tb,B) \right]
\ee
To see that \eqref{pippo} is equal to the r.h.s of \eqref{ldevcomb} one employs the substitution
$x_2 = 2x - \mathbb{E}(W)$. In doing so clearly the energetic contribution are equal
since
$$
-\frac{\tb}{2\E[W]}x_2^{2}  = -\frac{\tb }{2} \expec [W] - \frac{2\tb}{\expec [W]} x^2 + 2 \tb x 
$$
It remains to prove that 
$$
I(m,x_2) - \log(2)  =  {\tI_m}(x)
$$ 
This can be shown by changing the spin variables $\sigma_i$ to the variables $y_i=\frac 1 2 (\sigma_i +1)$ and introducing
	\be
	\hat{S}_{\sN}=\sum_{i\in [\sN]} y_{i},\quad \hat{S}^{\sss(w)}_{\sN}=\sum_{i\in [\sN]}  w_{i}y_{i}.\nn
	\ee
Observe that 
	\be
	{S}_{\sN}=2\hat{S}_{\sN}-\sN,\quad {S}^{\sss(w)}_{\sN}=2 \hat{S}^{\sss(w)}_{\sN}-\sN\,\E[W_{\sN}],\nn
	\ee
so that we can write
	\be
	\E_{P_{\sN}} [\exp ( t_1 S_{\sN}+ t_2 S^{\sss(w)}_{\sN})] = \exp(-\sN(t_1+ t_2 \E[W_{\sN}]))\,  \E_{P_{\sN}} [\exp (2 t_1 \hat{S}_{\sN}+2 t_2 \hat{S}^{\sss(w)}_{\sN})].\nn
	\ee
Since
	\be
	\E_{P_{\sN}} [\exp ( 2t_1 \hat{S}_{\sN}+ 2t_2 \hat{S}^{\sss(w)}_{\sN})] =  \E_{P_{\sN}} [ \Pi_{i\in [\sN]} \exp (2 t_1 + 2 w_i t_2)y_i  ] 
	=  \frac{1}{2^{\sN}} \Pi_{i\in [\sN]}  (1+\e^{2(t_1+w_i t_2)}),\nn
	\ee
we obtain that the moment generating function of $(S_{\sN},{S}^{\sss(w)}_{\sN})$ w.r.t.\ the product measure \eqref{proddelta} can be expressed as
	\be
	{c}_{\sN}({\bf t})= -\log 2 - (t_1+ t_2 \E[W_{\sN}])+\E[\log (1+\e^{2(t_1+W_n t_2)}) ].\nn
	\ee
Thus, arguing as in the proof  of Theorem~\ref{thm-ldp-weighted}, we obtain that the limit of ${c}_{\sN}({\bf t})$ exists and equals
	\be
	{c}({\bf t})=  -\log 2 - (t_1+ t_2 \E[W])+\E[\log (1+\e^{2(t_1+W t_2)}) ].\nn
	\ee
By applying the G\"artner-Ellis theorem we get the expression for the rate function
	\be
	{I}(x_1,x_2)= \sup_{(t_1,t_2)} \left ( t_1 x_1 + t_2 x_2  +\log 2 + (t_1+ t_2 \E[W])-\E[ \log (1+\e^{2(t_1+W t_2)})]\right ).\nn
	\ee
The stationarity conditions read as
	\be
	\label{eqlambda12bis}
	\left \{
	\begin{array}{l}
	\expec \left [   \frac{\displaystyle \e^{2(t_1 + W  t_2)}}{\displaystyle 1+ \e^{2(t_1 +W t_2)}}    \right ]  = \frac{\displaystyle 1+x_1}{\displaystyle 2},\\ \\
 	\expec \left [  W   \frac{\displaystyle \e^{2(t_1+ W t_2)}}{ \displaystyle 1+ \e^{2(t_1 + W t_2)}} \right ] = {\frac{\displaystyle \E[W] +x_2}{\displaystyle 2}}.
	\end{array}
	\right.
	\ee
Since $x_1$ represents the magnetization $m$ and $x_2$ represents  the weighted magnetization $m^{(w)}$, and using again the substitution $x = (x_2 + \mathbb{E}(W))/2$
we obtain that \eqref{eqlambda12bis} is identical to \eqref{eqlambda12} provided that  $\lambda_1$ is identified with  $2 t_2$ and $\lambda_2$  with  $2 t_1$.
}

\end{remark}


\begin{proof}[Proof of Theorem~\ref{thm-LDP-ICW}] Given a configuration  $\sigma$, we denote  by $n_+$ and $n_-$ the number of its positive resp. negative spins. 
We can group the spins in $\sigma$ according to either  $m_n$ or  $n_+$,  since these quantities are  related by  $n_+=\sN (1+m_{\sN})/2$. We can also identify each configuration of
spins $\sigma$ with the set $I_+\subset [\sN]$ of vertices in which $\sigma_i=1$, obviously the cardinality of this set is $|I_+|=n_+$. 
Given any $I_+ \subset [\sN]$, we  define
	\be
	q_k=\frac{1}{\sN}\# \{ i\in I_+ \mid w_i = a_k \}, \qquad  k=1,\ldots, K
	\ee
to be the frequency of type $a_k$ in $I_+$. Then
	\be
	\label{sumq}
	|I_+|=n_+= \sN \sum_{k=1}^K q_k
	\ee
and
	\be
	\label{adotq}
	r_n^{(w)}:=\frac{1}{\sN} \sum_{i \in I_+} w_i  = \sum_{k=1}^K a_k q_k \equiv  a \cdot q.
	\ee
Moreover, given $n$, we define the set
	\be
	{\cal Q}_n:=\left\{ \left (\frac{\ell_1}{n},\ldots,\frac{\ell_K}{n} \right) \mid  \ell_k\in  \bbN,\, \ell_k \le  \hat{\sN}_k (n),\,  k=1,\ldots, K \right\}\nn
	\ee
of the possible frequency vectors  $q=(q_1,q_2,\ldots, q_K)$.
\paragraph{\bf Exponential estimate for the conditional probability of $q=(q_1,q_2,\ldots, q_K)$.}
We start by counting the number of sets $I_+$ with $\lfloor \sN(1+m)/2\rfloor$ elements  and a given 
$q=(q_1,q_2,\ldots, q_K) \in {\cal Q}_n$ that satisfies the condition  $\sN \sum_{k=1}^K q_k= \lfloor \sN(\frac {1+m}{2})\rfloor =|I_+|=n_+$.  
For any $k=1,\ldots, K$, in $[n]$ there are $\sN \hat{p}_k$ sites corresponding to $a_k$, and we choose $\sN q_k$ out of them to form $I_+$. On the other hand, there are 
${\sN \choose n_+}\equiv {\sN\choose \lfloor \sN( {1+m})/{2} \rfloor }$ possible ways 
to form a set $I_+$ with $n_+$ elements. Thus, the conditional  distribution  of   $q=(q_1,q_2,\ldots, q_K)$ given $m$ is multi-hypergeometric, i.e.,
	\be\label{pob_q_cond_m}
	\prob_n(q_1,q_2,\ldots, q_K \mid m)=\frac{\displaystyle  \prod_{k=1}^K {\sN \hat{p}_k\choose \sN q_k}}{\displaystyle {\sN\choose \lfloor \sN(\frac {1+m}{2}) \rfloor  }}\, \indic{ q\in {\cal D}_n, \sum_{k=1}^K q_k= \frac 1 n \left \lfloor n \left ( \frac {1+m}{2} \right ) \right \rfloor}.
	\ee
The asymptotic behavior of this probability as $n\to \infty$, can be obtained by using  the Stirling's approximation $n! = \e^{-n}n^n \sqrt{2\pi n}(1+o(1)) $ to estimate of the binomial coefficient as
	\be
	{\sN b \choose  \sN a} = \e^{\sN[ b \log b - a \log a - (b-a) \log (b-a)]}\cdot  \frac{\sqrt{b}(1+o(1))}{\sqrt{a}\sqrt{b-a} \sqrt{2\pi \sN}},\nn	
	\ee
where $0<a<b$. Then, generalizing  the previous formula to a set of variables $a_k<b_k,\, k=1,\ldots, K$, we obtain

%
	\be\label{asymptnum}
	\prod_{k=1}^K {\sN b_k\choose \sN a_k} =  {\mathscr C}_n (a_k,b_k)(1+o(1)),
	\ee 
where 
	\be\label{asymptnumC}
	{\mathscr C}_n ({ a},{ b}): =  c_1\,  (2\pi n)^{- K/2} \exp \left (  \sN \sum_{k=1}^K [ b_k \log b_k - a_k \log a_k -(b_k - a_k)\log (b_k - a_k)] \right )
	\ee 
with $c_1= \prod_{k=1}^K \sqrt{\frac{b_k}{a_k(b_k-a_k)}}$ and the function is defined on the set 
	$$
	\{ (a_1,\ldots, a_K,b_1,\ldots, b_K) \in \R^{2K} |\, 0<a_k<b_k,\, k=1, \ldots, K\}.
	$$
We  now compute  the asymptotics of the numerator in \eqref{pob_q_cond_m}. Recalling that $ \hat{p}_k=p_k+ e_k$ and Taylor expanding the sum in \eqref{asymptnumC} and  $c_1$ as a function of  $b_k$'s, we obtain
	\begin{align}
 	\prod_{k=1}^K {\sN(p_k+ e_k)\choose \sN q_k} =&\, {\mathscr C}_n ({ q},{ p})\, [1+\sum_{k=1}^K c^{\sss(1)}_k e_k  + \sum_{k=1}^K O(  e^2_k)]\\
 	&\times  \exp \left (n \sum_{k=1}^K c^{\sss(2)}_k e_k(1+O(e_k))\right)(1+o(1)),\nn
	\end{align}
for some constants $ c^{\sss(1)}_k$ and $  c^{\sss(2)}_k$. The  second factor in the r.h.s.\ comes from the substitution $p_k  \to p_k+e_k$ in the factor $c_1$  of \eqref{asymptnumC},  and the third form the sum in the same equation.
From Condition \ref{cond-weightfinitetype} we have that these terms are both $(1+o(1))$. 
Then, we conclude that the numerator in \eqref{pob_q_cond_m}  is
\be
 \prod_{k=1}^K {\sN \hat{p}_k\choose \sN q_k} = {\mathscr C}_n ({ q},{ p}) (1+o(1)).
\ee
We can deal with the denominator in \eqref{pob_q_cond_m} in a similar fashion, obtaining:
	\be
	\label{asymptbin}
	{\sN\choose  \lfloor \sN (\frac {1+m}{2}) \rfloor } =c_2\, s_n (2\pi n)^{-\frac 1 2} \exp \left (\sN \left [ - \frac{1+m}{2} \log ( \frac{1+m}{2}) - \frac{1-m}{2} \log ( \frac{1-m}{2}) \right ] \right )(1+o(1)),
	\ee
with $c_2= \frac{2}{\sqrt{1-m^2}}$ and $s_n=s_n(m)=\exp [ (   n\frac{1+ m}{2} -  \lfloor n\frac{1+ m}{2} \rfloor) \log (\frac{1-m^2}{4})]$.
By plugging this estimate  and  \eqref{asymptnum}  in \eqref{pob_q_cond_m}, we finally obtain that
	\be
	\label{asymptp}
	\prob_n(q_1,q_2,\ldots, q_K \mid m) = \frac{c_1}{c_2}\, s_n\, (2\pi n)^{\frac{1-K}{2}} \e^{\sN g(q,m)} (1+o(1)),
	\ee
with
	\be
	\label{defg}
	g(q,m)= \left \{
	\begin{array}{ll}
	h(q,m), & \mbox{if}\; \; q_k \le p_k,\;  \sum_k q_k = \frac 1 n \lfloor  n (\frac{1+m}{2}) \rfloor,\\ \\
	-\infty, & \mbox{otherwise},
	\end{array}
	\right.
	\ee
and
	\be
	\label{hsmall}
	 h(q,m)=  \frac{1+m}{2} \log ( \frac{1+m}{2}) + \frac{1-m}{2} \log ( \frac{1-m}{2}) +\sum_{k=1}^K [ p_k \log p_k - q_k \log q_k -(p_k - q_k)\log (p_k - q_k)] .
	\ee
\paragraph{\bf Exponential estimate for the conditional probability of $r_n^{(w)}$.}
Let use introduce
	\be\nn
	{\cal A}_n(m)= \left \{ q_1 a_1+\cdots+q_K a_K\, |\,  q \in {\cal Q}_n,\, \sum_k q_k = \frac 1 n \left  \lfloor  n \left (\frac{1+m}{2} \right) \right \rfloor\right\},
	\ee
which are the sets of values of $r_n^{(w)}=\frac{1}{\sN} \sum_{i \in I_+} w_i $ corresponding to subsets  $I_+ \subset [n]$ with $  \lfloor  n(1+m)/2\rfloor $ elements.  

We have  that 
	\be
 	\frac{1}{2}(1+m)\, a_1 - \frac{\rho_n(m)}{n}   \le \inf  {\cal A}_n(m), \qquad   \sup  {\cal A}_n(m) \le  \frac{1}{2}(1+m)\, a_K - \frac{\rho_n(m)}{n},
	\label{bounds-Anm}
	\ee
with $\rho_n(m)=n (1+m)/2- \lfloor n (1+m)/2\rfloor $. Obviously  $\frac{\rho_n(m)}{n}=O(n^{-1})$, since $0\le \rho_n(m)<1$. 
Moreover, by \eqref{bounds-Anm} and the fact that $\inf  {\cal A}_n(m)\geq a_1 \sum_k q_k$ and $\sup  {\cal A}_n(m)\leq a_K \sum_k q_k$,
	\be
 	\inf  {\cal A}_n(m) \to  \frac{1}{2}(1+m)\, a_1, \quad \mbox \quad  \sup  {\cal A}_n(m) \to  \frac{1}{2}(1+m)\, a_K
	\ee
as $n \to \infty$.
The previous remark imply that   $r^{(w)}_n$ is close to some $x\in   \left [ \frac{1}{2}(1+m)\, a_1  ,  \frac{1}{2}(1+m)\, a_K  \right ]$ for large $n$. Therefore, we claim that 
	\be
	\label{probsumwi_bis}
	\lim_{\sN\to \infty} \frac{1}{\sN}  \log \prob_{\pi_\sN} \left (\sum_{i \in I_+} w_i  =  \sN\, x   ~\bigg|~  |I_+ | = \left \lfloor \sN(1+m)/2\right\rfloor \right) =
	\left \{ 
	\begin{array}{ll}
	\tS_m(x), & x\in   \left [ \frac{1}{2}(1+m)\, a_1  ,  \frac{1}{2}(1+m)\, a_K  \right ],\\ \\
	-\infty, & \mbox{otherwise},
	\end{array}
	\right.
	\ee 
where $\tS_m(x)$ has to be computed. 
To this end, we observe now that  the probability in \eqref{probsumwi_bis} can be written as 
	\be
 	\prob_{\pi_n} \left ( \sum_{i \in I_+} w_i  = \sN \, x ~\bigg|~ |I_+ | = \left \lfloor \sN {(1+m)}{/2} \right \rfloor \right )  
	= \sum_{q \in {\cal Q}_n\, \atop \,  a \cdot  q =x } \prob_n (q_1,q_2,\ldots, q_K \mid m) ,\nn
	\ee
where the sum is extended to those $k$-tuples $q \in  {\cal Q}_n $ for which the event $ \sum_{i \in I_+} w_i  =\sN\, x $ is realized. In the previous sum the term that corresponds to the larger value of the exponent $g(q,m)$ in \eqref{asymptp} controls the behavior  in the limit, the remaining terms being sub-leading.  The quantity  depending on $m$  in the definition  of $h(q,m)$,  see \eqref{hsmall}, is negative and the sum on $k$ is positive, while $h(q,m)$ is negative in the range  defined in the first line of \eqref{defg}.  Thus, defining
	\be
	\tilde{h}(q_1,q_2,\ldots, q_K)= \sum_{k=1}^K [ p_k \log p_k - q_k \log q_k -(p_k - q_k)\log (p_k - q_k)] ,\nn
	\ee
we have to find
	\be
	\mathscr{S}_n=\sup_{a\cdot q= x, \atop  \sum_k q_k=\frac 1 n \lfloor n (m+1)/2 \rfloor} \tilde{h}(q_1,q_2,\ldots, q_K).\nn
	\ee
In the previous equation the notation $\mathscr{S}_n$ emphasizes the fact that due to the constraints, the sup depends on $n$. As a consequence, the optimization point $q^\star=(q_1^\star,\ldots, q_k^\star)$ will depend on $n$. In order to find $q^\star$ we introduce the multipliers $\lambda_1$ and $\lambda_2$ conjugate to $x$ and $m$, and write the Lagrangian function as
	\be
	L(q_1,q_2,\ldots, q_K; \lambda_1,\lambda_2)= \tilde{h}(q_1,q_2,\ldots, q_K) +  \lambda_1 (\sum_{k=1}^K a_k q_k -x) + \lambda_2 ( \sum_k q_k-\widetilde{m}_n),\nn
	\ee
where we set
	\be\nn
	\widetilde{m}_n:= \frac 1 n \lfloor n (m+1)/2 \rfloor= \frac{1+m}{2} -\frac{\rho_n(m)}{n}=  \frac{1+m}{2} +O(n^{-1}).
	\ee
By imposing that   $\partial L/\partial q_k=0,\, k=1,\ldots, K$, we obtain that the stationarity  point $q^\star(n)=(q_1^\star(n),\ldots, q_k^\star(n))$ of the function $\tilde{h}$ satisfies
	\be\label{qstarn}
	q^\star_k(n)=\frac{p_k \e^{\lambda_1(n) a_k +\lambda_2(n)}}{ 1+ \e^{\lambda_1(n) a_k +\lambda_2(n)}}, \quad k=1,\ldots, K,\nn
	\ee
with $\lambda_1(n)=\lambda_1(x,\widetilde{m}_n)$,  $\lambda_2(n)=\lambda_2(x,\widetilde{m}_n)$.
By introducing the notation
	\be
	u_k(n)= \frac{q^\star_k(n)}{p_k}=\frac{\e^{\lambda_1(n) a_k +\lambda_2(n)}}{ 1+ \e^{\lambda_1(n) a_k +\lambda_2(n)}}\, ,\nn
	\ee	
we write 
	\begin{align}
	\mathscr{S}_n&=\tilde{h}(q^\star_1(n),q^\star_2(n),\ldots, q^\star_K(n))=-\sum_{k=1}^K p_k [u_k(n) \log u_k(n) + (1-u_k(n)) \log (1-u_k(n)) ]  \nonumber \\
	 &= \expec \left [  \frac{\e^{\lambda_1(n) W + \lambda_2(n)}}{1+\e^{\lambda_1(n) W + \lambda_2(n)} } \log \left ( \frac{\e^{\lambda_1(n) W + \lambda_2(n)}}{1+\e^{\lambda_1(n) W + \lambda_2(n)} } \right ) 
	+ \frac{1}{1+\e^{\lambda_1(n) W + \lambda_2(n)} } \log \left ( \frac{1}{1+\e^{\lambda_1(n) W + \lambda_2(n)} } \right )\right ].\label{sn}
	\end{align}
The relation between the multipliers $\lambda_1$, $\lambda_2$ and the parameters $x$, $\widetilde{m}_n$ can be made explicit by recalling that, since the probability vector $(p_1,\ldots, p_K)$ is the distribution of $W$, 
 from \eqref{adotq} we have
	\be\label{eqxn}
	x= \sum_{k=1}^K a_k p_k \frac{\e^{\lambda_1(n) a_k +\lambda_2(n)}}{ 1+ \e^{\lambda_1(n) a_k +\lambda_2(n)}}
	= \expec \left [  W   \frac{\e^{\lambda_1(n) W +\lambda_2(n)}}{ 1+ \e^{\lambda_1(n) W +\lambda_2(n)}} \right ],
	\ee
and, from \eqref{sumq},
	\be\label{eqmn}
	\widetilde{m}_n=\frac{1+m}{2}+O(n^{-1})=  \sum_{k=1}^K p_k  \frac{\e^{\lambda_1(n) a_k +\lambda_2(n)}}{ 1+ \e^{\lambda_1(n) a_k +\lambda_2(n)}}
	=  \expec \left [   \frac{\e^{\lambda_1 (n)W +\lambda_2(n)}}{ 1+ \e^{\lambda_1(n) W +\lambda_2(n)}}    \right ] ,
	\ee
where  \eqref{adotq} and \eqref{sumq} have been used.  By taking the limit of \eqref{eqxn} and \eqref{eqmn} as $n\to \infty$, we see that $\lambda_1(n)$ and $\lambda_2(n)$ converge to $\lambda_1$ and $\lambda_2$ that solve 
	\be
	x = \sum_{k=1}^K a_k p_k \frac{\e^{\lambda_1 a_k +\lambda_2}}{ 1+ \e^{\lambda_1 a_k +\lambda_2}}
	= \expec \left [  W   \frac{\e^{\lambda_1 W +\lambda_2}}{ 1+ \e^{\lambda_1 W +\lambda_2}} \right ],\nn
	\ee
and 
	\be
	\frac{1+m}{2}=  \sum_{k=1}^K p_k  \frac{\e^{\lambda_1 a_k +\lambda_2}}{ 1+ \e^{\lambda_1 a_k +\lambda_2}}
	=  \expec \left [   \frac{\e^{\lambda_1 W +\lambda_2}}{ 1+ \e^{\lambda_1 W +\lambda_2}}    \right ] ,\nn
	\ee	 
that is \eqref{eqlambda12}. From this fact it follows that in the same limit $n\to \infty$,
 	\be\nn
 	q^\star_k(n) \to q^\star_k =\frac{p_k \e^{\lambda_1 a_k +\lambda_2}}{ 1+ \e^{\lambda_1a_k +\lambda_2}}\, 
	\quad \mbox{and}\quad u_k(n)\to u_k=\frac{\e^{\lambda_1 a_k +\lambda_2}}{ 1+ \e^{\lambda_1 a_k +\lambda_2}}\, 
	\ee
and, thus,
 	\be
 	\mathscr{S}_n\to \mathscr{S}:=  \expec \left [  \frac{\e^{\lambda_1 W + \lambda_2}}{1+\e^{\lambda_1 W + \lambda_2} } 
	\log \left ( \frac{\e^{\lambda_1 W + \lambda_2}}{1+\e^{\lambda_1 W + \lambda_2} } \right ) 
	+ \frac{1}{1+\e^{\lambda_1 W + \lambda_2} } \log \left ( \frac{1}{1+\e^{\lambda_1 W + \lambda_2} } \right )\right ]
 	\ee
Then, from  \eqref{asymptp}, \eqref{hsmall}, and the previous display, we obtain the limit in \eqref{probsumwi_bis} with
	\begin{align}\label{ratefS}
	\tS_m(x)&= \frac{1+m}{2} \log ( \frac{1+m}{2}) + \frac{1-m}{2} \log ( \frac{1-m}{2})  \nonumber \\
	&\qquad+ \expec \left [  \frac{\e^{\lambda_1 W + \lambda_2}}{1+\e^{\lambda_1 W + \lambda_2} } \log \left ( \frac{\e^{\lambda_1 W + \lambda_2}}{1+\e^{\lambda_1 W + \lambda_2} } \right ) 
	+ \frac{1}{1+\e^{\lambda_1 W + \lambda_2} } \log \left ( \frac{1}{1+\e^{\lambda_1 W + \lambda_2} } \right )\right ].
	\end{align}
\paragraph{\bf Moment generating function of the Hamiltonian $H^{{\sss {\mathrm{ICW}}}}_{\sN}(\sigma)$.}
Our next step to compute  the cumulant generating function of the Hamiltonian \eqref{hamICWmwmn} that we rewrite as a function of
\be
\rnw= \frac{1}{\sN} \sum_{i\in I_+}  w_{i},\nn
\ee
for which we have proven \eqref{probsumwi_bis}. In this way we obtain
\be
H^{\sss {\mathrm{ICW}}}_{\sN}( \rnw ,m_n)= \sN \left [ \frac{ 2 \tb}{\expec [W]}  (\rnw)^2-2\tb  \frac{\expec{[W_{\sN}]}}{\expec [W]} \rnw + \frac{\tb}{2} \frac{\expec{[W_{\sN}]}^2}{\expec [W]} +  B m_n\right ].\nn
\ee
Now, writing $\expec[W_n] = \expec[W] + \epsilon_n$ and defining 
\be
 h^{\sss {\mathrm{ICW}}}_{\sN}( \rnw ,m_n):= \frac{2 \tb}{\expec [W]}  (\rnw)^2-2\tb \rnw + \frac{\tb}{2} \expec{[W]} +  B m_n ,\nn
\ee
we have
\be
H^{\sss {\mathrm{ICW}}}_{\sN}( \rnw ,m_n)= \sN\, h^{\sss {\mathrm{ICW}}}_{\sN}( \rnw ,m_n) + n\, \epsilon_n  \left [ - \frac{2 \tb}{\expec [W]} \rnw  +\tb  + \frac{\tb}{2 \expec [W]} \epsilon_n\right ].\nn
\ee
Since Condition \ref{cond-weightfinitetype} implies that $ \epsilon_n=o(1)$ the last addend  in the previous display is $o(n)$. Now we can finally write the 
cumulant generating function and  apply Varadhan's lemma  to compute
	\begin{align}\label{cumgenfhcond}
	&\lim_{\sN\to \infty } \frac{1}{\sN} \log  \expec_{\pi_n}\left [  \e^{H^{{\sss {\mathrm{ICW}}}}_{\sN}(\rnw, m_n)} \mid m_{\sN}=m\right ]
        =\lim_{\sN\to \infty } \frac{1}{\sN} \log  \expec_{\pi_n}\left [ \e^{ \sN\, [h^{\sss {\mathrm{ICW}}}_{\sN}( \rnw ,m_n) +o(1)]} \mid m_{\sN}=m\ \right ]\nonumber \\
	&\qquad= \frac{\tb }{2} \expec [W]+ B\, m + \sup_{x\in {\cal A}} \left [ \frac{2 \tb}{\expec [W]} x^2 -2 \tb  x - \tS_m(x) \right],
	\end{align}
where the large deviation property \eqref{probsumwi_bis} has been used and  ${\cal A}= \left [ \frac{1}{2}(1+m)\, a_1  ,  \frac{1}{2}(1+m)\, a_K  \right ]$. We can now move to the final step.

\paragraph{\bf Asymptotic behavior of $\mathbb{P}_{{\mu}^{\sss {\mathrm{ICW}}}_n}(m_{\sN}=m)$.}  Let us observe that, since the conditional  average on the left hand side to the previous display is computed with respect to the uniform measure $\pi_n(\sigma)=2^{-n}$ on the spins $\sigma$,
	\be
	\expec_{\pi_n} \left [  \e^{H^{\sss {\mathrm{ICW}}}_{\sN}(\sigma)} | m_{\sN}=m \right ]=\frac{\sum_\sigma \indic{m_{\sN}(\sigma)=m}  \e^{H^{\sss {\mathrm{ICW}}}_{\sN}(\sigma)}  \pi_n(\sigma)}{\mathbb{P}_{\pi_n}(m_{\sN}=m)}
	=\frac{ 2^{-n}\, Z^{\sss {\mathrm{ICW}}}_{\sN}}{\mathbb{P}_{\pi_n}(m_{\sN}=m) }\,\mathbb{P}_{{\mu}^{\sss {\mathrm{ICW}}}_n}(m_{\sN}=m),\nn
	\ee
with $ Z^{\sss {\mathrm{ICW}}}_{\sN}=\sum_\sigma  \e^{H^{\sss {\mathrm{ICW}}}_{\sN}(\sigma)}$, the partition function of the ICW model.
Thus,
	\be
	\frac{1}{\sN} \log \mathbb{P}_{{\mu}^{\sss {\mathrm{ICW}}}_n}(m_{\sN}=m)= \frac{1}{\sN} \log  \expec_{\pi_n} \left [  \e^{H^{\sss {\mathrm{ICW}}}_{\sN}(\sigma)} | m_{\sN}=m \right ]+ \frac{1}{\sN} \log {\mathbb{P}_{\pi_n}(m_{\sN}=m) } 
	-\frac{1}{\sN} \log Z^{\sss {\mathrm{ICW}}}_{\sN}\, + \log 2.\nn
	\ee
Since $ {\mathbb{P}_{\pi_n}(m_{\sN}=m) } =2^{-\sN} {\sN\choose \sN(\frac {1+m}{2}) }$, by \eqref{asymptbin},
	\be
	\lim_{\sN\to \infty} \frac{1}{\sN} \log {\mathbb{P}_{\pi_n}(m_{\sN}=m) } = -\log 2 - \frac{1+m}{2} \log ( \frac{1+m}{2}) - \frac{1-m}{2} \log ( \frac{1-m}{2}),\nn
	\ee
and
	\be
	\lim_{\sN\to \infty} \frac{1}{\sN} \log Z^{\sss {\mathrm{ICW}}}_{\sN}= \psi^{\sss {\mathrm{ICW}}} (\tb, B),\nn
	\ee
is the pressure of the Inhomogeneous Curie-Weiss model \cite{DomGiaGibHofPri16}. Thus, by \eqref{cumgenfhcond},
	\begin{align}
	\lim_{\sN\to\infty} \frac{1}{\sN} \log \mathbb{P}_{{\mu}^{\sss {\mathrm{ICW}}}_n}(m_{\sN}=m)= & \frac{\tb }{2} \expec [W]+ 
	B\, m + \sup_{x\in {\cal A}} \left [ \frac{2 \tb}{\expec [W]} x^2 -2 \tb x - \tS_m(x) \right]  -\psi^{\sss {\mathrm{ICW}}} (\tb, B) \nonumber \\
	&\quad  - \frac{1+m}{2} \log ( \frac{1+m}{2}) - \frac{1-m}{2} \log ( \frac{1-m}{2}),\nn
	\end{align}
from which, recalling \eqref{ratefS}, we obtain \eqref{ldevcomb} and \eqref{ratefIcomb}.
\end{proof}

\vspace{0.4cm}
\noindent
{\small
{\bfseries Acknowledgments.}
SD has been supported by the Deutsche Forschungsgemeinschaft (DFG) via RTG 2131 {\em High-dimensional Phenomena in Probability -- Fluctuations and Discontinuity}. 
We acknowledge financial support from the Italian Research Funding Agency (MIUR) through FIRB project 
``Stochastic processes in interacting particle systems: duality, metastability and their applications'', 
grant n.\ RBFR10N90W.
The work of RvdH is supported in part by the Netherlands
Organisation for Scientific Research (NWO) through VICI grant 639.033.806 and the Gravitation {\sc Networks} grant 024.002.003.
C. Giberti and C. Giardin\`a acknowledge financial supports from ``Fondo di Ateneo per la Ricerca 2015'' and ``Fondo di Ateneo per la Ricerca 2016'',
Universit\`a di Modena e Reggio Emilia. }

\end{document}